\def\UseBibLatex{1}

\ifx\JoCGVer\undefined
  \newcommand{\InJoCGVer}[1]{}
  \newcommand{\InNotJoCGVer}[1]{#1}
\else
  \newcommand{\InJoCGVer}[1]{#1}%
  \newcommand{\InNotJoCGVer}[1]{}%
\fi

   \documentclass[12pt]{article}%

   \usepackage[cm]{fullpage}%

\newcommand{\IfPrinterVer}[2]{#2}%
\providecommand{\Mh}[1]{{{#1}}}%

\IfFileExists{.latex_printer_friendly}{\def\GenPrinterVer{1}}{}%

\ifx\GenPrinterVer\undefined
   \IfFileExists{.latex_color}{\def\GenColorMath{1}}{} 
\else
   \renewcommand{\IfPrinterVer}[2]{#1}%
\fi

\ifx\GenColorMath\undefined
\else
   \renewcommand{\Mh}[1]{{\textcolor{red}{#1}}}%
\fi

\InJoCGVer{%
   \renewcommand{\Mh}[1]{{{#1}}}%
}

\usepackage{ifluatex}
\usepackage{ifxetex}

\ifluatex 
   \usepackage{fontspec}
   \usepackage[utf8]{luainputenc}
\else
   \ifxetex 
       \usepackage{fontspec}
   \else
       \usepackage[T1]{fontenc}
       \usepackage[utf8]{inputenc}
   \fi
\fi

\usepackage{caption}%
\usepackage{graphicx}%
\usepackage{amsmath}%
\usepackage{paralist}%
\usepackage{xcolor}%
\usepackage{xspace}%
\usepackage{euscript}%
\usepackage{picins}%
\usepackage{sariel_biblatex}%
\usepackage{stmaryrd}%
\usepackage{mleftright}%
\InNotJoCGVer{%
   \usepackage{footnote_fix}%
}

\usepackage{amsmath}%
\usepackage[amsmath,thmmarks]{ntheorem}
\theoremseparator{.}%
\usepackage{amssymb}

\usepackage{titlesec}
\titlelabel{\thetitle. }%

\newtheorem{theorem}{Theorem}[section] 
\newtheorem{lemma}[theorem]{Lemma}

\usepackage{color}
\usepackage{hyperref}%
\hypersetup{%
      breaklinks=true,%
      colorlinks=true,%
      urlcolor =[rgb]{0.250,0.0,0.0},
      linkcolor=[rgb]{0,0.245,0.00},%
      citecolor=[rgb]{0,0,0.245}
   }%

\newcommand{\Do}{{\small\bf do}\ }

\newcommand{\Return}{{\small\bf return\ }}

\newcommand{\If}{{\small\bf if}\ }
\newcommand{\Then}{{\small\bf then}\ }
\newcommand{\Else}{{\small\bf else}\ }

\newcommand{\For}{{\small\bf for}\ }

\newcommand{\xbeginlgox}{\begin{minipage}{1in}\begin{tabbing}
           \quad\=\qquad\=\qquad\=\qquad\=\qquad\=\qquad\=\qquad\=\kill}
        \newcommand{\xendlgox}{\end{tabbing}\end{minipage}}

\newenvironment{algorithmEnv}{
   \begin{tabular}{|l|}\hline\xbeginlgox\bigskip}
    {\xendlgox\\\hline\end{tabular}}

\newenvironment{program}{
   \begin{minipage}{4.0in}
   \begin{tabbing}
       \ \ \ \ \= \ \ \ \= \ \ \ \ \= \ \ \ \ \= \ \ \ \ \=
      \ \ \ \ \= \ \ \ \ \= \ \ \ \ \= \ \ \ \ \=
      \ \ \ \ \= \ \ \ \ \= \ \ \ \ \= \ \ \ \ \= \kill
}{
   \end{tabbing}
   \end{minipage}
}

\DefineNamedColor{named}{RedViolet} {cmyk}{0.07,0.90,0,0.34}

\newcommand{\AlgorithmI}[1]{{\textcolor[named]{RedViolet}{\texttt{\bf{#1}}}}}
\newcommand{\Algorithm}[1]{{\AlgorithmI{#1}\index{algorithm!#1@{\AlgorithmI{#1}}}}}

\theoremstyle{remark}%
\theoremheaderfont{\sf}
\theorembodyfont{\upshape}%

\newtheorem{defn}[theorem]{Definition}

\newtheorem*{remark:unnumbered}[theorem]{Remark}%
   \newtheorem{remark}[theorem]{Remark}%
   \newtheorem{remark:ext}[theorem]{Remark}%
   \newtheorem{example}[theorem]{Example}

   \theoremheaderfont{\em}%
   \theorembodyfont{\upshape}%
   \theoremstyle{nonumberplain}
   \theoremseparator{}
   \theoremsymbol{\myqedsymbol}
   \newtheorem{proof}{Proof:}

   \newcommand{\myqedsymbol}{\rule{2mm}{2mm}}

\numberwithin{figure}{section}%
\newcommand{\SarielThanks}[1]{\thanks{Department of Computer
      Science; 
      University of Illinois; 
      201 N. Goodwin Avenue;
      Urbana, IL, 61801, USA;
      {\tt sariel\atgen{}illinois.edu}; {\tt
         \url{http://sarielhp.org/}.} #1}}

\newcommand{\atgen}{\symbol{'100}}

\newcounter{tempEnumCounter}
\newcommand{\SaveEnumCounter}%
        {\setcounter{tempEnumCounter}{\value{enumi}}}%
\newcommand{\RestoreEnumCounter}%
        {\setcounter{enumi}{\value{tempEnumCounter}}}

\IfPrinterVer{%
  \definecolor{blue25}{rgb}{0,0,0}
  \DefineNamedColor{named}{RedViolet} {rgb}{0,0,0}%
}{%
  \definecolor{blue25}{rgb}{0,0,0.7}
  \DefineNamedColor{named}{RedViolet} {cmyk}{0.07,0.90,0,0.34}
}%

\InJoCGVer{%
  \definecolor{blue25}{rgb}{0,0,0}
  \DefineNamedColor{named}{RedViolet} {rgb}{0,0,0}%
}

\newcommand{\emphic}[2]{%
     \textcolor{blue25}{%
         \textbf{\emph{#1}}}%
         \index{#2}}
\newcommand{\emphi}[1]{\emphic{#1}{#1}}

\providecommand{\Mh}[1]{{{#1}}}

\newcommand{\pnt}{\Mh{p}}%
\newcommand{\pntA}{\Mh{q}}%
\newcommand{\pntB}{\Mh{u}}%
\newcommand{\pntC}{\Mh{x}}%
\newcommand{\pntD}{\Mh{y}}%

\newcommand{\si}[1]{#1}

\newcommand{\ray}{\sigma}
\newcommand{\rayA}{\sigma'}
\newcommand{\rayB}{\pi}

\newcommand{\SPath}{\mu}

\newcommand{\Path}{\tau}
\newcommand{\PathA}{\tau'}

\newcommand{\vertex}{v}

\newcommand{\AlgIsIn}{\Algorithm{isPntIn}\xspace}
\newcommand{\AlgIsInSubPoly}{\Algorithm{isInSubPoly}\xspace}

\newcommand{\Polygon}{\Mh{P}}%
\newcommand{\PolygonA}{\Mh{Q}}%

\newcommand{\cLenValue}{\Mh{8\pth{5 + \ceil{\ln \delta}}^2}}%

\newcommand{\permut}[1]{\left\langle {#1} \right\rangle}
\newcommand{\ClipPoly}[3]{#1\! \permut{ #2, #3 }}
\newcommand{\source}{\Mh{s}}

\newcommand{\target}{\Mh{t}}%

\renewcommand{\Re}{\mathbb{R}}

\newcommand{\solveVS}{\AlgorithmI{solve{VS}}\xspace}%
\newcommand{\WSet}{W}%

\newcommand{\cDim}{\Mh{\delta}}%
\newcommand{\VSpace}{\mathcal{V}}
\newcommand{\CSet}{\Mh{H}}%
\newcommand{\CSetA}{\Mh{X}}%
\newcommand{\CSetB}{\Mh{Y}}%
\newcommand{\CSetC}{\Mh{Z}}%
\newcommand{\vFuncC}{\Mh{\mathrm{cl}}}%
\newcommand{\vFuncX}[1]{\vFuncC\pth{#1}}%
\newcommand{\clX}[1]{\vFuncX{#1}}

\newcommand{\constB}{\Mh{\zeta}}%

\newcommand{\pDim}{\Mh{\phi}}%
\newcommand{\pDimB}{\Mh{\psi}}%
\newcommand{\ts}{\hspace{0.6pt}}

\newcommand{\algVio}{\AlgorithmI{violate}\xspace}%
\newcommand{\compBasis}{\AlgorithmI{comp{}Basis}\xspace}%

\newcommand{\Space}{\Mh{m}}%

\providecommand{\ceil}[1]{\left\lceil {#1} \right\rceil}
\providecommand{\floor}[1]{\left\lfloor {#1} \right\rfloor}

\providecommand{\brc}[1]{\left\{ {#1} \right\}}
\newcommand{\pth}[1]{\mleft({#1}\mright)}

\providecommand{\cardin}[1]{\left\lvert {#1} \right\rvert}

\newcommand{\edge}{e}

\newcommand{\disk}{D}
\renewcommand{\th}{th\xspace}

\newcommand{\CurveSet}{\Gamma}
\newcommand{\CurveSetA}{\Psi}
\newcommand{\CurveSetB}{\Upsilon}

\newcommand{\Corridor}{\Mh{C}}

\newcommand{\CLList}{L}
\newcommand{\Event}{\mathcal{E}}
\newcommand{\EventA}{\mathcal{F}}
\newcommand{\ISet}{\mathcal{I}}

\newcommand{\sep}[1]{\left|\, {#1} \right.}

\newcommand{\AllCorridors}{\EuScript{F}}%

\newcommand{\CD}[1]{\EuScript{C}\pth{ #1}}%

\newcommand{\DefSetChar}{\Mh{D}}%
\newcommand{\DefSet}[1]{\DefSetChar\pth{#1}}
\newcommand{\KillSet}[1]{\Mh{K}\pth{#1}}
\newcommand{\Sample}{\mathcal{S}}%

\newcommand{\curve}{\sigma}
\newcommand{\curveA}{\pi}

\newcommand{\Term}[1]{\textsf{#1}}
\newcommand{\LP}{\Term{LP}\index{linear programming}\xspace}
\newcommand{\NPComplete}{\Term{NPComplete}\xspace}%

\newcommand{\Basis}{\Mh{B}}%
\newcommand{\BasisA}{\Mh{{B}'}}%

\newcommand{\sopt}{\mu_{\mathrm{opt}}}

\newcommand{\node}{\Mh{u}}

\newcommand{\cell}{\Mh{\tau}}%
\newcommand{\cellX}[1]{\Mh{\tau}\pth{#1}}%
\newcommand{\cst}{\Mh{f}}%
\newcommand{\cstA}{\Mh{g}}%
\newcommand{\basisX}[1]{\Mh{\mathrm{basis}}\pth{#1}}%

\newcommand{\nA}{\Mh{\overline{{n}}}}%
\newcommand{\nB}{{\mathsf{t}}}

\newcommand{\remove}[1]{}

\newcommand{\RegionsSet}[1]{\Mh{\mathcal{F}}\pth{#1}}%

\newcommand{\Arr}{\Mh{\mathop{\mathrm{\EuScript{A}}}}}%
\newcommand{\ArrVD}[1]{\Mh{\Arr^{|}} \pth{#1}}

\newcommand{\Trap}{\Mh{\sigma}}

\newcommand{\IntRange}[1]{\left\llbracket #1 \right\rrbracket}

\newcommand{\ds}{\displaystyle}%
\newcommand{\XS}{\Mh{\mathcal{X}}}%

\newcommand{\probEL}[1]{\rho^{}_{\mathrm{0}}\pth{#1}}
\newcommand{\probEH}[2]{\rho^{}_{#1}\pth{#2}}
\newcommand{\BSetCh}{\mathcal{B}}%
\newcommand{\BasesAll}{\mathcal{B}}%
\newcommand{\BasisLE}[1]{{\BSetCh}_{\leq #1}}%
\newcommand{\RSample}{\Mh{R}}%
\newcommand{\etal}{\textit{et~al.}\xspace}

\newcommand{\Savron}{{\v S}avro{\v n}\xspace}
\newcommand{\PRG}{\Term{PRG}\xspace}

\newcommand{\Polytope}{\Mh{P}}%
\newcommand{\Vertices}{\Mh{V}}%
\newcommand{\Domain}{\Mh{\mathcal{D}}}

\newcommand{\Instance}{\Mh{I}}%
\newcommand{\vecI}{\Mh{v^{}_{\Instance}}}%
\newcommand{\TriX}[1]{\mathcal{T}\pth{#1}}

\providecommand{\Matousek}{Matou{\v s}ek\xspace}

\newcommand{\ListOfPrimes}{\begin{inparaenum}[(1)]
    \item $2$
    \item $5$

    \item $11$
    \item $17$
    \item $37$
    \item $67$
    \item $131$
    \item $257$
    \item $521$
    \item $1,031$
    \item $2,053$
    \item $4,099$
    \item $8,209$
    \item $16,411$
    \item $32,771$

    \item $65,537$
    \item $131,101$
    \item $262,147$
    \item $524,309$
    \item $1,048,583$
    \item $2,097,169$
    \item $4,194,319$
    \item $8,388,617$

    \item $16,777,259$
    \item $33,554,467$
    \item $67,108,879$
    \item $134,217,757$
    \item $268,435,459$
    \item $536,870,923$
    \item $1,073,741,827$
    \item $2,147,483,659$
    \item $4,294,967,311$
    \item $8,589,934,609$
    \item $17,179,869,209$
    \item $34,359,738,421$
    \item $68,719,476,767$
    \item $137,438,953,481$
    \item $274,877,906,951$
    \item $549,755,813,911$
    \item $1,099,511,627,791$
    \item $2,199,023,255,579$
    \item $4,398,046,511,119$
    \item $8,796,093,022,237$
    \item $17,592,186,044,423$
    \item $35,184,372,088,891$
    \item $70,368,744,177,679$
    \item $140,737,488,355,333$
    \item $281,474,976,710,677$
    \item $562,949,953,421,381$
    \item $1,125,899,906,842,679$.
\end{inparaitem}
}

\newcommand{\pbrcx}[1]{\mleft[ {#1} \mright]}
\newcommand{\Ex}[1]{\mathop{\mathbf{E}}\pbrcx{#1}}
\newcommand{\Prob}[1]{\mathop{\mathbf{Pr}}\pbrcx{#1}}

\newcommand{\HLinkShort}[2]{\hyperref[#2]{#1\ref*{#2}}}
\newcommand{\HLink}[2]{\hyperref[#2]{#1~\ref*{#2}}}
\newcommand{\HLinkPage}[2]{\hyperref[#2]{#1~\ref*{#2}%
      $_\text{p\pageref{#2}}$}}
\newcommand{\HLinkPageOnly}[1]{\hyperref[#1]{Page~\refpage*{#1}%
      $_\text{p\pageref{#1}}$}}

\newcommand{\HLinkSuffix}[3]{\hyperref[#2]{#1\ref*{#2}{#3}}}
\newcommand{\HLinkPageSuffix}[3]{\hyperref[#2]{#1\ref*{#2}%
      #3$_\text{p\pageref{#2}}$}}

\newcommand{\figlab}[1]{\label{fig:#1}}
\newcommand{\figref}[1]{\HLink{Figure}{fig:#1}}

\newcommand{\figrefpage}[1]{\HLinkPage{Figure}{fig:#1}}

\newcommand{\remlab}[1]{\label{rem:#1}}

\newcommand{\lemlab}[1]{\label{lemma:#1}}
\newcommand{\lemref}[1]{\HLink{Lemma}{lemma:#1}}%

\newcommand{\apndlab}[1]{\label{apnd:#1}}
\newcommand{\apndref}[1]{\HLink{Appendix}{apnd:#1}}

\newcommand{\thmlab}[1]{{\label{theo:#1}}}
\newcommand{\thmref}[1]{\HLink{Theorem}{theo:#1}}

\providecommand{\deflab}[1]{\label{def:#1}}
\newcommand{\defref}[1]{\HLink{Definition}{def:#1}}

\newcommand{\seclab}[1]{\label{sec:#1}}
\newcommand{\secref}[1]{\HLink{Section}{sec:#1}}

\newcommand{\exmlab}[1]{\label{example:#1}}
\newcommand{\exmref}[1]{\HLink{Example}{example:#1}}

\newcommand{\itemlab}[1]{\label{item:#1}}
\newcommand{\itemref}[1]{\HLinkSuffix{(}{item:#1}{)}}

\newcommand{\constPRG}{c'}

\newlength{\savedparindent}
\newcommand{\SaveIndent}{\setlength{\savedparindent}{\parindent}}
\newcommand{\RestoreIndent}{\setlength{\parindent}{\savedparindent}}

\IfFileExists{sariel_computer.sty}{\def\sarielComp{1}}{}
\ifx\sarielComp\undefined%
\newcommand{\SarielComp}[1]{}
\newcommand{\NotSarielComp}[1]{#1}%
 \else
\newcommand{\SarielComp}[1]{#1}%
\newcommand{\NotSarielComp}[1]{}%
\fi

\BibLatexMode{\addbibresource{sublinear_space}}

\begin{document}

\InNotJoCGVer{%
   \title{Shortest Path in a Polygon using Sublinear Space%
      \footnote{%
         Work on this paper was partially supported by NSF AF awards
         CCF-1421231 and 
         CCF-1217462.  
         A preliminary version of this paper appeared in SoCG 2015
         \cite{h-sppus-15}.%
      }%
   }%
}%
\InJoCGVer{%
   \title{\MakeUppercase{Shortest Path in a Polygon using Sublinear
         Space}%
      \thanks{%
         Work on this paper was partially supported by a NSF AF awards
         CCF-1421231, and 
         CCF-1217462.  
         A preliminary version of this paper appeared in SoCG 2015
         \cite{h-sppus-15}.%
      }%
   }%
}

\author{Sariel Har-Peled\SarielThanks{}}

\date{\today}

\maketitle

\InNotJoCGVer{%
   \setfnsymbol{stars}
}

\begin{abstract}
    We resolve an open problem due to Tetsuo Asano, showing how to
    compute the shortest path in a polygon, given in a read only
    memory, using sublinear space and subquadratic time. Specifically,
    given a simple polygon $\Polygon$ with $n$ vertices in a read only
    memory, and additional working memory of size $\Space$, the new
    algorithm computes the shortest path (in $\Polygon$) in
    $O( n^2 /\, \Space )$ expected time, assuming
    $\Space = O(n /\log^2 n)$. This requires several new tools, which
    we believe to be of independent interest.

    Specifically, we show that violator space problems, an abstraction
    of low dimensional linear-programming (and \LP-type problems), can
    be solved using constant space and expected linear time, by
    modifying Seidel's linear programming algorithm and using
    pseudo-random sequences.
\end{abstract}

\section{Introduction}

Space might not be the final frontier in the design of algorithms but
it is an important constraint. Of special interest are algorithms that
use sublinear space. Such algorithms arise naturally in streaming
settings, or when the data set is massive, and only a few passes on
the data are desirable. Another such setting is when one has a
relatively weak embedded processor with limited high quality
memory. For example, in 2014, flash memory could withstand around
100,000 rewrites before starting to deteriorate. Specifically, imagine
a hybrid system having a relatively large flash memory, with
significantly smaller RAM. That is to a limited extent the setting in
a typical smart-phone\footnote{For example, a typical smart-phone in
   2014 has 2GB of RAM and 16GB of flash memory. I am sure these
   numbers would be laughable in a few years. So it goes.}.

\paragraph*{The model.} %
The input is provided in a read only memory, and it is of size $n$. We
have $O(\Space)$ available space which is a read/write space (i.e.,
the \emph{work space}). We assume, as usual, that every memory cell is
a word, and such a word is large enough to store a number or a
pointer.  We also assume that the input is given in a reasonable
representation\footnote{In some rare cases, the ``right'' input
   representation can lead directly to sublinear time algorithms. See
   the work by Chazelle \etal \cite{clm-sga-05}.}.  

Since a memory cell has $\Omega(\log n)$ bits, for $m=O(1)$, this is
roughly the $\log$-space model in complexity.  An example of such
algorithms are the standard \NPComplete reductions, which can all be
done in this model. Algorithms developed in this model of limited work
memory include:
\begin{inparaenum}[(i)]
    \item median selection \cite{mr-sroms-96},
    \item deleting a connected component in a binary image
    \cite{a-ipaec-12},
    \item triangulating a set of points in the plane, computing their
    Voronoi diagram, and Euclidean minimum spanning tree
    \cite{amrw-cwsag-11}, to name a few.
\end{inparaenum}
For more details, see the introduction of Asano \etal
\cite{abbkm-mcasp-13, abbkm-mcasp-14}.

\paragraph*{The problem.} %
We are given a simple polygon $\Polygon$ with $n$ vertices in the
plane, and two points $\source, \target \in \Polygon$ -- all provided
in a read-only memory. We also have $O(\Space)$ additional read-write
memory (i.e., work space). The task is to compute the shortest path
from $\source$ to $\target$ inside $\Polygon$.

Asano \etal \cite{abbkm-mcasp-13} showed how to solve this problem, in
$O(n^2/\Space)$ time, using $O(\Space)$ space. The catch is that their
solution requires quadratic time preprocessing. In a talk by Tetsuo
Asano, given in a workshop in honor of his 65\th birthday (during SoCG
2014), he posed the open problem of whether this quadratic
preprocessing penalty can be avoided. This work provides a positive
answer to this question.

\bigskip%
\noindent%
\begin{minipage}{0.7\linewidth}
    \textbf{If linear space is available.} %
    The standard algorithm \cite{lp-esppr-84} for computing the
    shortest path in a polygon, triangulates the polygon,
    (conceptually) computes the dual graph of the triangulation, which
    yields a tree, with a unique path between the triangles that
    contains the source $\source$, and the target $\target$. This path
    specifies the sequence of diagonals crossed by the shortest path,
    and it is now relatively easy to walk through this sequence of
    triangles and maintain the shortest paths from the source to the
    two endpoints of each diagonal. These paths share a prefix path,
    and then diverge into two concave chains (known together as a
    \emph{funnel}). Once arriving to the destination, one computes the
    unique tangent from the destination $\target$ to one of these
    chains, and the (unique) path, formed by the prefix together with
    the tangent, defines the shortest path, which can be now extracted
    in linear time. See figure on the right.
\end{minipage}%
\begin{minipage}{0.29\linewidth}
    \hfill \includegraphics[width=0.93\linewidth,height=4.5cm]%
    {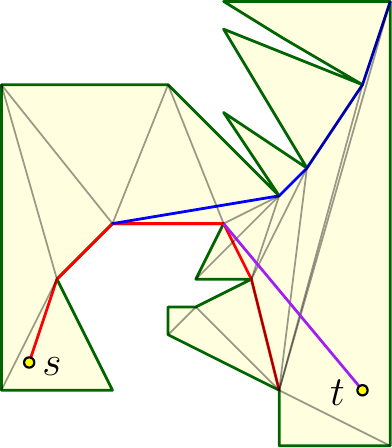}
\end{minipage}

\paragraph*{Sketch of the new algorithm.} %
The basic idea is to decompose the polygon into bigger pieces than
triangles. Specifically, we break the polygon into canonical pieces
each of size $O(\Space)$. To this end, we break the given polygon
$\Polygon$ into $\ceil{n / \Space}$ polygonal chains, each with at
most $ \Space$ edges. We refer to such a chain as a \emph{curve}. We
next use the notion of \emph{corridor decomposition},
introduced\footnote{%
   Many somewhat similar decomposition ideas can be found in the
   literature. For example, the decomposition of a polygon into
   monotone polygons so that the pieces, and thus the original
   polygon, can be triangulated \cite{bcko-cgaa-08}.  Closer to our
   settings, Kapoor and Maheshwari \cite{km-eaesp-88} use the dual
   graph of a triangulation to create a decomposition of the free
   space among polygonal obstacles into corridors.  Nevertheless, the
   decomposition scheme of \cite{h-qssp-14} is right for our nefarious
   purposes, but the author would not be surprised if it was known
   before. Well, at least this footnote is new!} %
by the author \cite{h-qssp-14}, to (conceptually) decompose the
polygon into canonical pieces (i.e., corridors). Oversimplifying
somewhat, each edge of the medial axis \cite{b-tends-67} corresponds
to a corridor, which is a polygon having portions of two of the input
curves as floor and ceiling, and additional two diagonals of
$\Polygon$ as gates.  It is relatively easy, using constant space and
linear time, to figure out for such a diagonal if it separates the
source from the destination. Now, start from the corridor containing
the source, and figure out which of its two gates the shortest path
goes through. We follow this door to the next corridor and continue in
this fashion till we reach the destination. Assuming that computing
the next corridor can be done in roughly linear time, this algorithm
solves the shortest path problem in $O(n^2/\Space)$ time, as walking
through a corridor takes (roughly) linear time, and there are
$O(n /\Space)$ pieces the shortest path might go through. (One also
needs to keep track of the funnel being constructed during this walk,
and prune parts of it away because of space considerations.)

\paragraph*{point location queries in a canonical decomposition.} %
To implement the above, we need a way to perform a point location
query in the corridor decomposition, without computing it explicitly.
Specifically, we have an implicit decomposition of the polygon into
subpolygons, and we walk through this decomposition by performing a
sequence of point location queries on the shared borders between these
pieces.

More generally, we are interested in any canonical decomposition that
partition the underlying space into cells. Such a partition is induced
by a set of objects, and every cell is defined by a constant number of
objects. Standard examples of such partitions are (i) vertical
decomposition of segments in the plane, or (ii) bottom vertex
triangulation of the Voronoi diagram of points in $\Re^3$. Roughly
speaking, any partition that complies with the Clarkson-Shor framework
\cite{c-arscg-88,cs-arscg-89} is such a canonical decomposition, see
\secref{p:l:violator:spaces} for details.

If space and time were not a constraint, we could build the
decomposition explicitly, Then a standard point location query in the
history DAG would yield the desired cell. Alternatively, one can
perform this point location query in the history DAG implicitly,
without building the DAG beforehand, but it is not obvious how to do
so with limited space. Surprisingly, at least for the author, this
task can be solved using techniques related to low-dimensional linear
programming.

\paragraph*{Violator spaces.} %
Low dimensional linear programming can be solved in linear time
\cite{m-lpltw-84}. Sharir and Welzl \cite{sw-cblpr-92} introduced
\emph{\LP-type} problems, which are an extension of linear
programming. Intuitively, but somewhat incorrectly, one can think
about \LP-Type algorithms as solving low-dimensional convex
programming, although Sharir and Welzl \cite{sw-cblpr-92} used it to
decide in linear time if a set of axis-parallel rectangles can be
stabbed by three points (this is quite surprising as this problem has
no convex programming flavor). \LP-type problems have the same notions
as linear programming of bases, and an objective function. The
function scores such bases, and the purpose is to find the basis that
does not violate any constraint and minimizes (or maximizes) this
objective. A natural question is how to solve such problems if there
is no scoring function of the bases.

This is captured by the notion of \emph{violator spaces}
\cite{r-tpmlc-07, s-amop-07, gmrs-vssa-06, gmrs-vssa-08,
   bg-cavs-11}. The basic idea is that every subset of constraints is
mapped to a unique basis, every basis has size at most $\cDim$
($\cDim$ is the dimension of the problem, and is conceptually a
constant), and certain conditions on consistency and monotonicity
hold. Computing the basis of a violator space is not as easy as
solving \LP-type problems, because without a clear notion of progress,
one can cycle through bases (which is not possible for \LP-type
problems). See \Savron{} \cite{s-amop-07} for an example of such
cycling. Nevertheless, Clarkson's algorithm \cite{c-lvali-95} works
for violator spaces \cite{bg-cavs-11}.

We revisit the violator space framework, and show the following:
\begin{compactenum}[(A)]%
    \smallskip%
    \item Because of the cycling mentioned above, the standard version
    of Seidel's linear programming algorithm \cite{s-sdlpc-91} does
    not work for violator spaces. However, it turns out that a variant
    of Seidel's algorithm does work for violator spaces.

    \smallskip%
    \item We demonstrate that violator spaces can be used to solve the
    problem of point location in canonical decomposition. While in
    some cases this point location problem can be stated as an
    \LP-type problem, stating it as a violator space problem seems to
    be more natural and elegant.
    
    \smallskip%
    \item The advantages of Seidel's algorithm is that except for
    constant work space, the only additional space it needs is to
    store the random permutations it uses. We show that one can use
    pseudo-random generators (\PRG{}s) to generate the random
    permutation, so that there is no need to store it explicitly. This
    is of course well known -- but the previous analysis
    \cite{m-cgitr-94} for linear programming implied only that the
    expected running time is $O(n \log^{\cDim-1} n)$, where $\cDim$ is
    the combinatorial dimension. Building on Mulmuley's work
    \cite{m-cgitr-94}, we do a somewhat more careful analysis, showing
    that in this case one can use backward analysis on the random
    ordering of constraints generated, and as such the expected
    running time remains linear. 
    
    \smallskip%
    This implies that one can solve violator space problems (and thus,
    \LP and \LP-type problems) in constant dimension, using constant
    space, in expected linear time. 

    \smallskip%
    This compares favorably with previous work:
    \begin{inparaenum}[(A)]
        \item The algorithm for \LP by Mulmuley
        \cite{m-cgitr-94}, mentioned above, that uses constant space,
        and takes $O(n \log^{\cDim-1} n )$ expected time.
        \item Chan and Chen's \cite[Theorem 3.12]{cc-mpga-07} algorithm
        for \LP and \LP-type problems that uses $O(\log n)$ space and
        takes expected linear time.
    \end{inparaenum}
\end{compactenum}

\paragraph*{Paper organization.} %
We present the new algorithm for computing the basis of violator
spaces in \secref{violator:spaces}.  The adaptation of the algorithm
to work with constant space is described in
\secref{v:s:constant:space}.  We describe corridor decomposition and
its adaptation to our setting in \secref{corridor:decomp}.  We present
the shortest path algorithm in \secref{shortest:path}.


\section{Violator spaces and constant space %
   algorithms}
\seclab{violator:spaces}

First, we review the formal definition of \emph{violator spaces}
\cite{r-tpmlc-07, s-amop-07, gmrs-vssa-06, gmrs-vssa-08,
   bg-cavs-11}. We then show that a variant of Seidel's algorithm for
linear programming works for this abstract setting, and show how to
adapt it to work with constant space and in expected linear time.

\subsection{Formal definition of violator space}

Before delving into the abstract framework, let us consider the
following concrete example -- hopefully it would help the reader in
keeping track of the abstraction.

\begin{example}
    \exmlab{p:l:v:d}%
    We have a set $\CSet$ of $n$ segments in the plane, and we would
    like to compute the vertical trapezoid of $\ArrVD{\CSet}$ that
    contains, say, the origin, where $\ArrVD{\CSet}$ denote the
    vertical decomposition of the arrangement formed by the segments
    of $\CSet$.  For a subset $\CSetA \subseteq \CSet$, let
    $\cellX{\CSetA}$ be the vertical trapezoid in $\ArrVD{\CSetA}$
    that contains the origin. The vertical trapezoid $\cellX{\CSetA}$
    is defined by at most four segments, which are the \emph{basis} of
    $\CSetA$. A segment $\cst \in \CSet$ \emph{violates}
    $\cell = \cellX{\CSetA}$, if it intersects the interior of
    $\cellX{\CSetA}$. The set of segments of $\CSet$ that intersect
    the interior of $\cell$, denoted by $\vFuncX{\cell}$ or
    $\vFuncX{\CSetA}$, is the \emph{conflict list} of $\cell$.
\end{example}

Somewhat informally, a violator space identifies a vertical trapezoid
$\cell = \cellX{\CSetA}$, by its conflict list $\vFuncX{\CSetA}$, and
not by its geometric realization (i.e., $\cell$).

\begin{defn}
    \deflab{violator:space}%
    A \emphi{violator space} is a pair $\VSpace = (\CSet, \vFuncC)$,
    where $\CSet$ is a finite set of \emphi{constraints}, and
    $\vFuncC:2^\CSet \rightarrow 2^\CSet$ is a function, such that:
    \begin{compactitem}
        \item \textbf{Consistency}: For all $\CSetA \subseteq \CSet$,
        we have that $\vFuncX{\CSetA} \cap \CSetA = \emptyset$.

        \item \textbf{Locality}: For all
        $\CSetA \subseteq \CSetB \subseteq \CSet$, if
        $\vFuncX{\CSetA} \cap \CSetB = \emptyset$ then
        $\vFuncX{\CSetA} = \vFuncX{\CSetB}$.

        \item \textbf{Monotonicity}: For all
        $\CSetA \subseteq \CSetB \subseteq \CSetC \subseteq \CSet$, if
        $\vFuncX{\CSetA} = \vFuncX{\CSetC}$ then
        $\vFuncX{\CSetA} = \vFuncX{\CSetB} = \vFuncX{\CSetC}$.
    \end{compactitem}

    A set $\Basis \subseteq \CSetA \subseteq \CSet$ is a \emphi{basis}
    of $\CSetA$, if $\vFuncX{\Basis} = \vFuncX{\CSetA}$, and for any
    proper subset $\BasisA \subset \Basis$, we have that
    $\vFuncX{\BasisA} \neq \vFuncX{\Basis}$. The \emphi{combinatorial
       dimension}, denoted by $\cDim$, is the maximum size of a basis.
\end{defn}

Note that consistency and locality implies monotonicity. For the sake
of concreteness, it is also convenient to assume the following (this
is strictly speaking not necessary for the algorithm).
\begin{defn}
    For any $\CSetA \subseteq \CSet$ there is a unique \emphi{cell}
    $\cellX{\CSetA}$ associated with it, where for any
    $\CSetA, \CSetB \subseteq \CSet$, we have that if
    $\vFuncX{\CSetA} \neq \vFuncX{\CSetB}$ then
    $\cellX{\CSetA} \neq \cellX{\CSetB}$.  Consider any
    $\CSetA \subseteq \CSet$, and any $\cst \in \CSet$.  For
    $\cell = \cellX{\CSetA}$, the constraint $\cst$ \emphi{violates}
    $\cell$ if $\cst \in \vFuncX{\CSetA}$ (or alternatively, $\cst$
    \emph{violates} $\CSetA$).
\end{defn}

Finally, we assume that the following two basic operations are
available:
\begin{compactitem}
    \item \algVio{}($\cst, \Basis$): Given a basis $\Basis$ (or its
    cell $\cell = \cellX{\Basis}$) and a constraint $\cst$, it returns
    true $\iff$ $\cst$ violates $\cell$.

    \item \compBasis{}($\CSetA$): Given a set $\CSetA$ with at most
    $(\cDim+1)^2$ constraints, this procedure computes
    $\basisX{\CSetA}$, where $\cDim$ is the combinatorial dimension of
    the violator space\footnote{We consider $\basisX{\CSetA}$ to be
       unique (that is, we assume implicitly that the input is in
       general position). This can be enforced by using
       lexicographical ordering, if necessary, among the defining
       bases always using the lexicographically minimum basis.}. For
    $\cDim$ a constant, we assume that this takes constant time.
\end{compactitem}

\subsubsection{Linear programming via violator spaces}

Consider an instance $\Instance$ of linear programming in $\Re^d$ --
here the \LP is defined by a collection of linear inequalities with
$d$ variables.  The instance $\Instance$ induces a polytope
$\Polytope$ in $\Re^d$, which is the \emph{feasible domain} --
specifically, every inequality induces a halfspace, and their
intersection is the polytope.

The following interpretation of the feasible polytope is somewhat
convoluted, but serves as a preparation for the next example.  The
vertices $\Vertices$ of the polytope $\Polytope$ induce a
triangulation (assuming general position) of the sphere of directions,
where a direction ${v}$ belongs to a vertex $\pnt$, if and only if
$\pnt$ is an extreme vertex of $\Polytope$ in the direction of
$v$. Now, the objective function of $\Instance$ specifies a direction
$\vecI$, and in solving the \LP, we are looking for the extreme vertex
of $\Polytope$ in this direction.

Put differently, every subset $\CSet$ of the constraints of
$\Instance$, defines a triangulation $\TriX{\CSet}$ of the sphere of
directions. So, let the cell of $\CSet$, denoted by
$\cell = \cellX{\CSet}$, be the spherical triangle in this
decomposition that contains $\vecI$.  The basis of $\CSet$ is the
subset of constraints that define $\cellX{\CSet}$.  A constraint
$\cst$ of the \LP violates $\cell$ if the vertex induced by the basis
$\basisX{\CSet}$ (in the original space), is on the wrong side of
$\cst$.

Thus solving the \LP instance $\Instance = (\CSet,\vecI)$ is no more
than performing a point location query in the spherical triangulation
$\TriX{\CSet}$, for the spherical triangle that contains $\vecI$.

\begin{remark:unnumbered}
    An alternative, and somewhat more standard, way to see this
    connection is via geometric duality \cite[Chapter 25]{h-gaa-11}
    (which is \emph{not} \LP duality). The duality maps the upper
    envelope of hyperplanes to a convex-hull in the dual space (i.e.,
    we assume here that in the given \LP all the halfspaces (i.e.,
    each constraint corresponds to a halfspace) contain the positive
    ray of the $x_d$-axis -- otherwise, we need to break the
    constraints of the \LP into two separate families, and apply this
    reduction separately to each family). Then, extremal query on the
    feasible region of the \LP, in the dual, becomes a vertical ray
    shooting query on the upper portion of the convex-hull of the dual
    points -- that is, a point location query in the projection of the
    triangulation of the upper portion of the convex-hull of the dual
    points.
\end{remark:unnumbered}

\subsubsection{Point location via violator spaces}%
\seclab{p:l:violator:spaces}

\exmref{p:l:v:d} hints to a more general setup. So consider a space
decomposition into canonical cells induced by a set of objects. For
example, segments in the plane, with the canonical cells being the
vertical trapezoids. More generally, consider any decomposition of a
domain into simple canonical cells induced by objects, which falls
under the Clarkson-Shor framework \cite{c-arscg-88,cs-arscg-89}
(defined formally below). Examples of this include point location in a
\begin{inparaenum}[(i)]
    \item Delaunay triangulation,
    \item bottom vertex triangulation in an arrangement of
    hyperplanes,
    \item and vertical decomposition of arrangements of curves and
    surfaces in two or three dimensions, to name a few.
\end{inparaenum}


\begin{defn}[Clarkson-Shor framework \cite{c-arscg-88, cs-arscg-89}]%
    \deflab{c:s:framework}%
    Let $\Domain$ be an underlying domain, and let $\CSet$ be a set of
    objects, such that any subset $\RSample \subseteq \CSet$
    decomposes $\Domain$ into canonical cells
    $\RegionsSet{\RSample}$. This decomposition complies with the
    Clarkson-Shor framework, if we have that for every cell $\cell$
    that arises from such a decomposition has a \emph {defining set}
    $\DefSet{ \cell } \subseteq \CSet$. The size of such a defining
    set is assumed to be bounded by a constant $\cDim$.
    The \emph{stopping set} (or conflict list) $\KillSet{ \cell }$ of
    $\cell$ is the set of objects of $\CSet$ such that including any
    object of $\KillSet{ \cell }$ in $\RSample$ prevents $\cell$ from
    appearing in $\RegionsSet{\RSample }$.  We require that for any
    $\RSample \subseteq \CSet$, the following conditions hold:
    \begin{compactenum}[\qquad(i)]
    \item \itemlab{axiom:1} For any $\cell \in \RegionsSet{\RSample}$,
      we have $\DefSet{\cell} \subseteq \RSample$
      and $\RSample \cap \KillSet{\cell} = \emptyset$.
    \item \itemlab{axiom:2} If $\DefSet{\cell} \subseteq \RSample$
      and $\KillSet{ \cell }\cap \RSample=\emptyset$,
      then $\cell\in \RegionsSet{\RSample}$.
    \end{compactenum}
\end{defn}

For a detailed discussion of the Clarkson-Shor framework, see
Har-Peled \cite[Chapter 8]{h-gaa-11}.  We next provide a quick example
for the reader unfamiliar with this framework.

\begin{figure}
    \centerline{\includegraphics[scale=0.8]{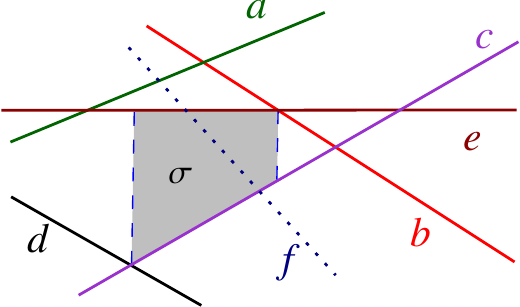}}
    \caption{Vertical decomposition and the Clarkson-Shor
       framework. The defining set of $\Trap$ is
       $\DefSet{\Trap} = \brc{a,b,c,d}$, and its stopping set is
       $\KillSet{\Trap} = \brc{f}$. }
    \figlab{stopping:set}
\end{figure}
\begin{example}
    Consider a set of segments $\CSet$ in the plane.  For a subset
    $\RSample \subseteq \CSet$, let $\ArrVD{\RSample}$ denote the
    \emphi{vertical decomposition} of the plane formed by the
    arrangement $\Arr(\RSample)$ of the segments of $\RSample$. This
    is the partition of the plane into interior disjoint vertical
    trapezoids formed by erecting vertical walls through each vertex
    of $\Arr(\RSample)$. Here, each object is a segment, a region is a
    {vertical trapezoid}, and $\RegionsSet{\RSample}$ is the set of
    vertical trapezoids in $\ArrVD{\RSample}$.  Each trapezoid $\Trap
    \in \RegionsSet{\RSample}$ is defined by at most four segments (or
    lines) of $\RSample$ that define the region covered by the
    trapezoid $\Trap$, and this set of segments is
    $\DefSet{\Trap}$. Here, $\KillSet{\Trap}$ is the set of segments
    of $\CSet$ intersecting the interior of the trapezoid $\Trap$ (see
    \figref{stopping:set}).
\end{example}

\begin{lemma}
    Consider a canonical decomposition of a domain into simple cells,
    induced by a set of objects, that complies with the Clarkson-Shor
    framework \cite{c-arscg-88, cs-arscg-89}. Then, performing a point location
    query in such a domain is equivalent to computing a basis of a
    violator space.
\end{lemma}

\begin{proof}
    This follows readily from definition, but we include the details
    for the sake of completeness. We use the notations of
    \defref{c:s:framework} above.

    So consider a fixed point $\pnt \in \Domain$, and the task at hand
    is to compute the cell $\cell \in \RegionsSet{\CSet}$ that
    contains $\pnt$ (here, assuming general position implies that
    there is a unique such cell). In particular, for
    $\RSample \subseteq \CSet$, we define $\Basis = \basisX{\RSample}$
    to be the defining set of the cell $\cell$ of
    $\RegionsSet{\RSample}$ that contains $\pnt$, and the conflict
    list to be $\vFuncX{\RSample} = \KillSet{\cell}$.

    We claim that computing $\basisX{\CSet}$ is a violator space
    problem. We need to verify the conditions of
    \defref{violator:space}, which is satisfyingly easy. Indeed,
    consistency is condition \itemref{axiom:1}, and locality is
    condition \itemref{axiom:2} in \defref{c:s:framework}. (We remind
    the reader that consistency and locality implies monotonicity, and
    thus we do not have to verify that it holds.)%
\end{proof}

It seems that for all of these point location problems, one can solve
them directly using the \LP-type technique. However, stating these
problems as violator space instances is more natural as it avoids the
need to explicitly define an artificial ordering over the bases, which
can be quite tedious and not immediate (see \apndref{l:p:type} for an
example).


\begin{figure}[t]
    \centerline{%
       \begin{algorithmEnv}%
           \solveVS{}$\pth{\WSet, \CSetA \bigr. }\Bigr.$:%
           \+\+\\
           $\permut{\cst_1,\ldots, \cst_m}$: A random permutation of
           the constraints of $\CSetA$.\\
           $\Basis_0 \leftarrow \compBasis\pth{ \WSet}$\\
           \For $i=1$ to $m$ \Do\+\\
           \If \algVio{}($\cst_i, \Basis_{i-1}$) \Then\+\\
           $\Basis_i \leftarrow \solveVS\pth{ \WSet \cup \Basis_{i-1}
              \cup \brc{\cst_i},\, \brc{\cst_1,\ldots,
                 \cst_i}\bigr. }$
           \-\\
           \Else \+\\
           $\Basis_i \leftarrow \Basis_{i-1}$
           \-\-\\
           \Return $\Basis_m$
       \end{algorithmEnv}%
    }%
    \caption{The algorithm for solving violator space problems.  The
       parameter $\WSet$ is a set of $O(\cDim^2)$ witness constraints,
       and $\CSetA$ is a set of $m$ constraints.  The function return
       $\basisX{ \WSet \cup \CSetA}$.  To solve a given violator
       space, defined implicitly by the set of constraints $\CSet$,
       and the functions \algVio and \compBasis, one calls
       \solveVS{}$\pth{ \brc{}, \CSet }$.}
    \figlab{V:S:algorithm}
\end{figure}

\subsection{The algorithm for computing the basis of %
   a violator space}

The input is a violator space $\VSpace = (\CSet, \vFuncC)$
with $n = \cardin{\CSet}$ constraints, having combinatorial dimension
$\cDim$.

\subsubsection{Description of the algorithm}

The algorithm is a variant of Seidel's algorithm \cite{s-sdlpc-91} --
it picks a random permutation of the constraints, and computes
recursively in a randomized incremental fashion the basis of the
solution for the first $i$ constraints. Specifically, if the $i$\th
constraint violates the basis $\Basis_{i-1}$ computed for the first
$i-1$ constraints, it calls recursively, adding the constraints of
$\Basis_{i-1}$ and the $i$\th constraint to the set of constraints
that must be included whenever computing a basis (in the recursive
calls).  The resulting code is depicted in \figref{V:S:algorithm}.

The only difference with the original algorithm of Seidel, is that the
recursive call gets the set
$\WSet \cup \Basis_{i-1} \cup \brc{\cst_i}$ instead of
$\basisX{\Basis_{i-1} \cup \brc{\cst_i}}$ (which is a smaller set).
This modification is required because of the potential cycling between
bases in a violator space.

\subsubsection{Analysis}

The key observation is that the depth of the recursion of \solveVS{}
is bounded by $\cDim$, where $\cDim$ is the combinatorial dimension of
the violator space. Indeed, if $\cst_i$ violates a basis, the
constraints added to the witness set $\WSet$ guarantee that any
subsequent basis computed in the recursive call contains $\cst_i$, as
testified by the following lemma.

\begin{lemma}
    \lemlab{witness}%
    Consider any set $\CSetA \subseteq \CSet$. Let
    $\Basis = \basisX{\CSetA}$, and let $\cst$ be a constraint in
    $\CSet \setminus \CSetA$ that violates $\Basis$. Then, for any
    subset $\CSetB$ such that
    \begin{math}
        \Basis \cup \brc{\cst} \subseteq \CSetB \subseteq \CSetA \cup
        \brc{\cst},
    \end{math}
    we have that $\cst \in \basisX{\CSetB}$.
\end{lemma}

\begin{proof}
    Assume that this is false, and let $\CSetB$ be the bad set with
    $\BasisA = \basisX{\CSetB}$, such that
    $\cst \notin \BasisA$.  Since $\cst \in \CSetB$, by consistency,
    $\cst \notin \vFuncX{\CSetB}$, see \defref{violator:space}.  By
    definition $\vFuncX{\CSetB} = \vFuncX{\BasisA}$, which implies
    that $\cst \notin \vFuncX{\BasisA}$; that is, $\cst$ does not
    violate $\BasisA$.

    Now, by monotonicity, we have
    \begin{math}
        \vFuncX{\CSetB}%
        =%
        \vFuncX{\CSetB \setminus \brc{\cst}}%
        =%
        \vFuncX{\BasisA}.
    \end{math}

    By assumption, $\Basis \subseteq \CSetB \setminus \brc{\cst}$,
    which implies, again by monotonicity, as
    $\Basis \subseteq \CSetB \setminus \brc{\cst} \subseteq \CSetA$,
    that
    \begin{math}
        \vFuncX{\CSetA}%
        = %
        \vFuncX{\CSetB \setminus \brc{\cst}} =
        \vFuncX{\Basis},
    \end{math}
    as $\Basis = \basisX{\CSetA}$. But that implies that
    \begin{math}
        \vFuncX{\Basis}%
        =%
        \vFuncX{\CSetB \setminus \brc{\cst}}%
        =%
        \vFuncX{\BasisA}.
    \end{math}
    As $\cst \notin \vFuncX{\CSetB}$, this implies that $\cst$ does
    not violate $\Basis$, which is a contradiction.
\end{proof}

\begin{lemma}
    \lemlab{d:depth}%
    The depth of the recursion of \solveVS{}, see
    \figrefpage{V:S:algorithm}, is at most $\cDim$, where $\cDim$ is
    the combinatorial dimension of the given instance.
\end{lemma}

\begin{proof}
    Consider a sequence of $k$ recursive calls, with
    \begin{math}
        \WSet_0%
        \subseteq%
        \WSet_1%
        \subseteq%
        \WSet_2 \subseteq \cdots \subseteq \WSet_k
    \end{math}
    as the different values of the parameter $\WSet$ of \solveVS{},
    where $\WSet_0 = \emptyset$ is the value in the top-level
    call. Let $\cstA_j$, for $j=1,\ldots,k$, be the constraint whose
    violation triggered the $j$\th level call. Observe that
    $\cstA_{j} \in \WSet_j$, and as such all these constraints must be
    distinct (by consistency).  Furthermore, we also included the
    basis $\Basis_j'$, that $\cstA_j$ violates, in the witness set
    $\WSet_j$. As such, we have that   
    \begin{math}
        \WSet_j = \brc{ \cstA_1, \ldots, \cstA_j } \cup \Basis_1' \cup
        \cdots \cup \Basis_j'.
    \end{math}
    By \lemref{witness}, in any basis computation done inside the
    recursive call $\solveVS\pth{ \WSet_j, \ldots }$, it must be that
    $\cstA_j \in \basisX{\WSet_t}$, for any $t \geq j$.  As such, we
    have $\cstA_1, \ldots, \cstA_k \in \basisX{\WSet_k}$. Since a
    basis can have at most $\cDim$ elements, this is possible only if
    $k \leq \cDim$, as claimed.
\end{proof}

\begin{theorem}
    \thmlab{v:s:n:space}
    Given an instance of violator space 
    $\VSpace = (\CSet, \vFuncC)$ with $n$ constraints, and
    combinatorial dimension $\cDim$, the algorithm 
    \noindent\solveVS{}$\pth{ \emptyset, \CSet }$, see 
    \figref{V:S:algorithm}, computes $\basisX{\CSet}$.  The expected
    number of violation tests performed is bounded by
    $O(\cDim^{\cDim+1} n)$.  Furthermore, the algorithm performs in
    expectation $O\pth{ \pth{ \cDim\ln n }^\cDim }$ basis computations
    (on sets of constraints that contain at most $\cDim(\cDim+1)$
    constraints).

    In particular, for constant combinatorial dimension $\cDim$, with
    violation test and basis computation that takes constant time,
    this algorithm runs in $O(n)$ expected time.
\end{theorem}

\begin{proof}
    \lemref{d:depth} implies that the recursion tree has bounded
    depth, and as such this algorithm must terminate.  The correctness
    of the result follows by induction on the depth of the recursion.
    By \lemref{d:depth}, any call of depth $\cDim$, cannot find any
    violated constraint in its subproblem, which means that the
    returned basis is indeed the basis of the constraints specified in
    its subproblem. Now, consider a recursive call at depth
    $j < \cDim$, which returns a basis $\Basis_{m}$ when called on the
    constraints $\permut{\cst_1, \ldots, \cst_m}$. The basis
    $\Basis_{m}$ was computed by a recursive call on some prefix
    $\permut{\cst_1,\ldots,\cst_i}$, which was correct by (reverse)
    induction on the depth, and none of the constraints
    $\cst_{i+1}, \ldots, \cst_m$ violates $\Basis_m$, which implies
    that the returned basis is a basis for the given subproblem. Thus,
    the result returned by the algorithm is correct.

    Let $C_k(m)$ be the expected number of basis computations
    performed by the algorithm when run at recursion depth $k$, with
    $m$ constraints. We have that
    $C_{\cDim}(m) =1$, and
    \begin{math}
        C_k(m) = 1+ \Ex{\Bigl. \sum_{i=1}^m X_i C_{k+1}(i)},
    \end{math}
    where $X_i$ is an indicator variable that is one, if and only if
    the insertion of the $i$\th constraint caused a recursive call. We
    have, by backward analysis, that
    \begin{math}
        \Ex{X_i} = \Prob{X_i=1} \leq \min\pth{ {\cDim}/{i}, 1}.
    \end{math}
    As such, by linearity of expectations, we have
    \begin{math}
        C_k(m) \leq 1+ \sum_{i=1}^m \min\pth{ \frac{\cDim}{i}, 1}
        C_{k+1}(i).
    \end{math}
    As such, we have
    \begin{align*}
      C_{\cDim-1}(m)%
      \leq%
      1+ \sum_{i=1}^m \min\pth{ \frac{\cDim}{i},
      1} C_{\cDim}(i)%
      =%
      \cDim +  \sum_{i=\cDim}^m  \frac{\cDim}{i}%
      \leq%
      \cDim  + \cDim \ln(m/\cDim)
      \leq%
      \cDim \ln(m),
    \end{align*}
    assuming $\cDim \geq e$.  We conclude that
    \begin{math}
        C_{\cDim-1}(m) \leq \cDim \ln (m)
    \end{math}
    and
    $C_0(m) = O\pth{ \pth{ \cDim\ln m }^\cDim }$.

    As for the expected number of violation tests, a similar analysis
    shows that
    \begin{math}
        V_\cDim(m) = m%
    \end{math}
    and 
    \begin{math}
        V_{\cDim - j}(m) = m+ \sum_{i=1}^m \min\pth{ \frac{\cDim}{i},
           1} V_{\cDim - (j -1)}(i).
    \end{math} %
    As such, assuming inductively that
    $V_{\cDim - (j-1)}(m) \leq j \cDim^{j-1} m$, we have that
    \begin{align*}
      V_{\cDim - j}(m)%
      &
        \leq%
        m+ \sum_{i=1}^m  \frac{\cDim}{i}
        j\cDim^{j-1} i
        \leq%
        m+ j \cDim^{j} m
        \leq%
        ( j+1) \cDim^{j} m,
    \end{align*} %
    implying that $V_0(m) = O(\cDim^{\cDim+1} m)$, which also bounds
    the running time.
\end{proof}%

\begin{remark}
    While the constants in \thmref{v:s:n:space} are not pretty, we
    emphasize that the $O$ notation in the bounds do not hide
    constants that depends on $\cDim$.
\end{remark}

\subsection{Solving the violator space problem with %
   constant space and linear time}
\seclab{v:s:constant:space}

The idea for turning \solveVS{} into an algorithm that uses little
space, is observing that the only thing we need to store (implicitly)
is the random permutation used by \solveVS{}.

\subsubsection{Generating a random permutation %
   using pseudo-random generators}

To avoid storing the permutation, one can use pseudo-random techniques
to compute the permutation on the fly. For our algorithm, we do not
need a permutation - any random sequence that has uniform distribution
over the constraints and is sufficiently long, would work.
 
\begin{lemma}
    \lemlab{p:r:g}%
    Fix an integer $\pDim> 0$, a prime integer $n$, and an integer
    constant $\constPRG \geq 12$. One can compute a random sequence of
    numbers
    $X_1, \ldots, X_{\constPRG n} \in \IntRange{n} = \brc{1,\ldots,
       n}$, such that:
    \begin{compactenum}[\qquad(A)]
        \item The probability of $X_i =j$ is $1/n$, for any
        $i \in \IntRange{n}$ and $j \in \IntRange{ \constPRG n}$.
        \item The sequence is $\pDim$-wise independent.
        \item Using $O(\constPRG \pDim)$ space, given an index $i$,
        one can compute $X_i$ in $O(\pDim)$ time.
    \end{compactenum}
\end{lemma}

\begin{proof}
    This is a standard pseudo-random generator (\PRG) technique,
    described in detail by Mulmuley \cite[p.  399]{m-cgitr-94}.  We
    outline the idea. Randomly pick $\pDim+1$ coefficients
    $\alpha_0, \ldots, \alpha_{\pDim} \in \brc{0, \ldots, n-1}$
    (uniformly and independently), and consider the random polynomial
    $f(x)= \sum_{i=0}^{\pDim} \alpha_i x^i$, and set
    $p(x)= (f(x) \mod n)$.  Now, set $X_i = 1 + p(i)$, for
    $i=1,\ldots, n$. It is easy to verify that the desired properties
    hold. To extent this sequence to be of the desired length, pick
    randomly $\constPRG$ such polynomials, and append their sequences
    together to get the desired longer sequence.  It is easy to verify
    that the longer sequence is still $\pDim$-wise independent.
\end{proof}

We need the following concentration result on $2k$-wise independent
sequences -- for the sake completeness, we provide a proof at
\apndref{s:p:s}.

\newcommand{\LemmaCehbychevKBody}{%
   Let $X_1, \ldots, X_n$ be $n$ random indicator variables that are
   $2\pDimB$-wise independent, where $p = \Prob{X_i=1} \geq 1/n$, for
   all $i$. Let $\Bigl.Y = \sum_i X_i$, and let $\mu = \Ex{Y} = p n$.
   Then, we have that
    \begin{math}        
         \Prob{\bigl. Y = 0 }%
        \leq%
        \Prob{\bigl. \cardin{Y - \mu} \geq \mu}%
        \leq%
        \pth{4\pDimB^2 /\mu \bigr. }^\pDimB / \sqrt{2\pDimB}.
    \end{math}
}

\begin{lemma}[Theorem A.2.1, \cite{m-cgitr-94}]%
    \lemlab{chebychev:k}%
    \LemmaCehbychevKBody
\end{lemma}

The following lemma testifies that this \PRG sequence, with good
probability, contains the desired basis (as such, conceptually, we can
think about it as being a permutation of $\IntRange{n}$).

\begin{lemma}
    \lemlab{p:r:g:failure}%
    Let $\Basis \subseteq \IntRange{n}$ be a specific set of $\cDim$
    numbers.  Fix an integer $\pDim \geq 8 + 2\cDim$, and let
    $\XS = \permut{ X_1, \ldots, X_{\constPRG n}}$ be a $\pDim$-wise
    independent random sequence of numbers, each uniformly distributed
    in $\IntRange{n}$, where $\constPRG$ is any constant
    $\geq \cLenValue$.  Then, the probability that the elements of
    $\Basis$ do not appear in $\XS$ is bounded by, say, $1/20$.
\end{lemma}

\begin{proof}
    Fix an element $b$ of $\Basis$.  Let $Y_i$ be an indicator
    variable that is one $\iff$ $X_i = b$. Observe that
    $p = \Prob{Y_i=1} = 1/n$. As such, we have that
    \begin{math}
        \mu = \Ex{ \sum_{i=1}^{\constPRG n} Y_i } = \constPRG = 8 k^2,
    \end{math}
    for $k = 5 + \ceil{\ln \cDim}$.

    As $k \leq \pDim/2$, we can interpret $\XS$ as a $2k$-wise
    independent sequence.  By \lemref{chebychev:k}, we have that the
    probability that $b$ does not appear in $\XS$ is bounded by
    \begin{math}
        \Prob{\sum_i Y_i = 0 }%
        \leq%
        \pth{ 4k^2 / \mu }^{k}/ \sqrt{2k}.%
    \end{math}
    We want this probability to be smaller than $1/(20\cDim)$, Since
    $\mu =8k^2 $, this is equivalent to
    \begin{math}
        \pth{ 4k^2 / \mu }^{k} \cDim \leq {\sqrt{2k}}/{20}%
        \iff%
        \cDim \leq 2^k {\sqrt{2k}}/{20}%
    \end{math}
    which holds as $2^k \geq 32 \cDim$.

    Now, by the union bound, the probability that any element of
    $\Basis$ does not appear in $\XS$ is bounded by
    ${\cDim}/(20 {\cDim}) = {1}/{20}$.
\end{proof}%

\begin{remark:ext}
    \remlab{technicalities}%
    There are several low level technicalities that one needs to
    address in using such a \PRG sequence instead of a truly random
    permutation:
    \begin{compactenum}[(A)]
        \item \textsc{Repeated numbers are not a problem}: the
        algorithm \solveVS (see \figrefpage{V:S:algorithm}) ignores a
        constraint that is being inserted for the second time, since
        it cannot violate the current basis.
        
        \item \itemlab{adding:missing}%
        \textsc{Verifying the solution}: The sequence (of the indices)
        of the constraints used by the algorithm would be
        $X_1, \ldots, X_{\constPRG n}$. This sequence might miss some
        constraints that violate the computed solution.
        
        As such, in a second stage, the algorithm checks if any of the
        constraints $1,2,\ldots, n$ violates the basis computed.  If a
        violation was found, then the sequence generated failed, and
        the algorithm restarts from scratch -- resetting the \PRG used
        in this level, regenerating the random keys used to initialize
        it, and rerun it to generate a new sequence.
        \lemref{p:r:g:failure} testifies that the probability for such
        a failure is at most $1/20$ -- thus restarting is going to be
        rare.

        \item \textsc{Independence between levels}: We will use a
        different \PRG for each level of the recursion of
        \solveVS. Specifically, we generate the keys used in the \PRG
        in the beginning of each recursive call.  Since the depth of
        the recursion is $\cDim$, that increases the space requirement
        by a factor of $\cDim$.
        
        \item \textsc{If the subproblem size is not a prime}: In a
        recursive call, the number of constraints given (i.e., $m$)
        might not be a prime. To this end, the algorithm can store
        (non-uniformly\footnote{That is, sufficiently large sequence
           of such numbers have to be hard coded into the program, so
           that it can handle any size input the program is required
           to handle. This notion is more commonly used in the context
           of nonuniform circuit complexity.}), a list of primes, such
        that for any $m$, there is a prime $m'\geq m$ that is at most
        twice bigger than $m$\footnote{Namely, the program hard codes
           a list of such primes. The author wrote a program to
           compute such a list of primes, and used it to compute $50$
           primes that cover the range all the way to $10^{15}$ (the
           program runs in a few seconds). This would require hard
           coding $O(\log^2 n)$ bits storing these
           primes. Alternatively, one can generate such primes on the
           fly (and no hard coding of bits is required), but the
           details become more complicated, and this issue is somewhat
           orthogonal to the main trust of the paper.

           Nevertheless, here is a short outline of the idea:
           Legendre's conjectured that for any integer $x > 0$, there
           is a prime in the interval $[x, x +O(\sqrt{x})]$ (Cram\'er
           conjectured that this interval is of length $O(\log^2 x)$).
           As such, since one can test primality for a number $x$ in
           time polynomial in the number of bits encoding $x$, it
           follows that one can find a prime close to $2^i$ in time
           which is $o(2^i)$, which is all we need for our scheme.%
        }. Then the algorithm generates the sequence modulo $m'$, and
        ignores numbers that are larger than $m$. This implies that
        the sequence might contain invalid numbers, but such numbers
        are only a constant fraction of the sequence, so ignoring them
        does not change the running time analysis of our
        algorithm. (More precisely, this might cause the running time
        of the algorithm to deteriorate by a factor of
        $\exp(O(\cDim))$, but as we consider $\cDim$ to be a constant,
        this does not effect our analysis.)
    \end{compactenum}
\end{remark:ext}

One needs now to prove that backward analysis still works for our
algorithm for violator spaces. The proof of the following lemma is
implied by a careful tweaking of Mulmuley's analysis -- we provide the
details in \apndref{backward}.

\newcommand{\LemmaBAnalysisBody}{%
   Consider a violator space $\VSpace = (\CSet, \vFuncC)$ with
   $n =\cardin{\CSet}$, and combinatorial dimension $\cDim$. Let
   $\XS = X_1, X_2, \ldots, $ be a random sequence of constraints of
   $\CSet$ generated by a $\pDim$-wise independent distribution (with
   each $X_i$ having a uniform distribution), where
   $\pDim > 6\cDim + 9$ is a constant.  Then, for $i > 2\cDim$, the
   probability that $X_i$ violates
   $\Basis = \basisX{X_1, \ldots, X_{i-1}}$ is $O(1/i)$.  }

\begin{lemma}[See \apndref{backward}]%
       \lemlab{backward:a:works}%
    \LemmaBAnalysisBody
\end{lemma}

\begin{remark:unnumbered}
    \lemref{backward:a:works} states that backwards analysis works on
    pseudo-random sequences when applied to the Clarkson-Shor
    \cite{c-arscg-88, cs-arscg-89} settings, where the defining set of the event
    (or basis) of interest has constant size. However, there are cases
    when one would like to apply backwards analysis when the defining
    set might have an unbounded size, and the above does not hold in
    this case. 

    For an example where a defining set might have unbounded size,
    consider a random permutation of points, and the event being a
    point being a vertex of the convex-hull of the points inserted so
    far. This is used in a recent proof showing that the complexity of
    the convex-hull of $n$ points, picked uniformly and randomly, in
    the hypercube in $d$ dimensions has $O(\log^{d-1} n)$ vertices,
    with high probability \cite{chr-fpuvp-15},
\end{remark:unnumbered}

\subsubsection{The result}

\begin{theorem}
    \thmlab{v:s:c:space}
    Given an instance of a violator space $\VSpace = (\CSet, \vFuncC)$
    with $n$ constraints, and combinatorial dimension $\cDim$, one can
    compute $\basisX{\CSet}$ using $O\pth{\cDim^2 \log^2 \cDim }$
    space. For some constant $\constB = O\pth{ \cDim \log^2 \cDim}$,
    we have that:
    \begin{compactenum}[(A)]
        \item The expected number of basis computations is
        $O\pth{ \bigl. \smash{ \pth{ \constB \ln n }^\cDim} }$, each
        done over $O( \cDim^2 )$ constraints.
        
        \item The expected number of violation tests performed is
        $O\bigl( \constB^\cDim n \bigr)$.

        \item The expected running time (ignoring the time to do the
        above operations) is $O\bigl( \constB^\cDim n \bigr)$.
    \end{compactenum}
\end{theorem}

\begin{proof}
    The algorithm is described above. As for the analysis, it follows
    by plugging the bound of \lemref{backward:a:works} into the proof
    of \thmref{v:s:n:space}.

    The only non-trivial technicality is to handle the case that the
    \PRG sequence fails to contain the basis. Formally, abusing
    notations somewhat, consider a recursive call on the constraints
    indexed by $\IntRange{n}$, and let $\Basis$ be the desired basis
    of the given subproblem.  By \lemref{p:r:g:failure}, the
    probability that $\Basis$ is not contained in the generated \PRG
    is bounded by $1/20$ -- and in such a case the sequence has to be
    regenerated till success.  As such, in expectation, this has a
    penalty factor of (say) $2$ on the running time in each
    level. Overall, the analysis holds with the constants
    deteriorating by a factor of (at most) $ 2^\cDim$.
\end{proof}

\begin{remark}
    Note, that the above pseudo-random generator technique is well
    known, but using it for linear programming by itself does not make
    too much sense. Indeed, pseudo-random generators are sometimes
    used as a way to reduce the randomness consumed by an
    algorithm. That in turn is used to derandomize the
    algorithm. However, for linear programming Megiddo's original
    algorithm was already linear time deterministic. Furthermore,
    Chazelle and \Matousek \cite{cm-ltdao-96}, using different
    techniques showed that one can even derandomize Clarkson's
    algorithm and get a linear running time with a better constant.

    Similarly, using \PRG{}s to reduce space of algorithms is by now a
    standard technique in streaming, see for example the work by Indyk
    \cite{i-sdpge-06}, and references therein.
\end{remark}



\section{Corridor decomposition}
\seclab{corridor:decomp}

\begin{figure}[t]
    \begin{center}
        \begin{tabular}{c%
          cc}
          \includegraphics[page=2,width=0.3\linewidth]{figs/medial_axis_new}%
          &
            \includegraphics[page=4,width=0.3\linewidth]{figs/medial_axis_new}%
          &%
            \includegraphics[page=5,width=0.3\linewidth]%
            {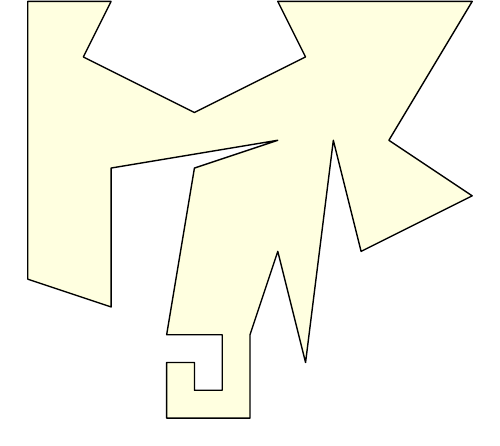}\\[-0.1cm]
          (i) & (ii) & (iii) 
        \end{tabular}
        %
    \end{center}

    \vspace{-0.5cm}

    \caption{An example of a corridor decomposition for a polygon: %
       (i) Input curves, %
       medial axis and active vertices, %
       (ii) their critical circles, %
       and their spokes, and %
       (iii) the resulting corridor decomposition.%
    }
    \figlab{corridor:decomposition:polygon}
\end{figure}

\subsection{Construction}

The decomposition here is similar to the decomposition described by
the author in a recent work \cite{h-qssp-14}.

\begin{defn}[Breaking a polygon into curves]
    \deflab{breakup}%
    Let the given polygon $\Polygon$ have the vertices
    $\vertex_1, \ldots \vertex_n$ in counterclockwise order along its
    boundary. Let $\curve_i$ be the polygonal curve having the
    vertices
    \begin{math}
        \vertex_{(i-1)\Space + 1}, \vertex_{(i-1)\Space+2}, %
        \linebreak[1] \vertex_{(i-1)\Space+3}, \ldots,
        \vertex_{i\Space + 1},
    \end{math}
    for $i =1, \ldots, \nA-1$, where
    \begin{align*}
        \nA = \floor{ (n -1) / \Space}+1.
    \end{align*}
    The last polygonal curve is
    $\sigma_{\nA} = \vertex_{(\nA-1)\Space + 1},
    \vertex_{(\nA-1)\Space + 2}, \ldots \vertex_n, \vertex_1$.
    Note, that given $\Polygon$ in a read only memory, one can encode
    any curve $\curve_i$ using $O(1)$ space. Let
    $\CurveSet = \brc{ \curve_1, \ldots, \curve_{\nA}}$ be the
    resulting set of polygonal curves. From this point on, a
    \emphi{curve} refers to a polygonal curve generated by this
    process.
\end{defn}

\subsubsection{Corridor decomposition for the whole polygon} %
Next, consider the medial axis of $\Polygon$ restricted to the
interior of $\Polygon$ (see Chin \etal~\cite{csw-fmasp-99} for the
definition and an algorithm for computing the medial axis of a
polygon).  A vertex $\vertex$ of the medial axis corresponds to a disk
$\disk$, that touches the boundary of $\Polygon$ in two or three
points (by the general position assumption, that the medial axis of
the polygon does not have degenerate vertices, not in any larger
number of points). The medial axis has the topological structure of a
tree.

To make things somewhat cleaner, we pretend that there is a little
hole centered at every vertex of the polygon if it is the common
endpoint of two curves. This results in a medial axis edge that comes
out of the vertex as an angle bisector, both for an acute angle (where
a medial axis edge already exists), and for obtuse angles, see
\figref{corridor:decomposition:polygon} and
\figref{corridor:decomposition:polygon:2}.

\parpic[r]{%
   \begin{minipage}{5cm}%
       \includegraphics[width=0.98\linewidth]{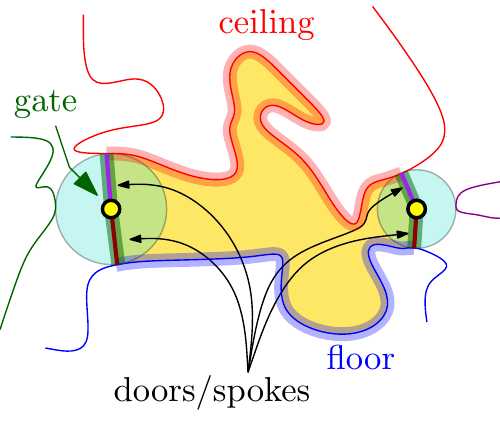}%
       \captionof{figure}{A corridor and its doors.}  \figlab{doors}%
   \end{minipage}%
}%

A vertex of the medial axis is \emphi{active} if its disk touches
three different curves of $\CurveSet$. It is easy to verify that there
are $O( \nA )$ active vertices. The segments connecting an active
vertex to the three (or two) points of tangency of its empty disk with
the boundary of $\Polygon$ are its \emphi{spokes} or
\emphi{doors}. Introducing these spokes breaks the polygon into the
desired \emphi{corridors}.

Such a corridor is depicted on the right. It has a \emphi{floor} and a
\emphi{ceiling} -- each is a subcurve of some input curve. In
addition, each corridor might have four additional doors grouped into
``double'' doors. Such a double door is defined by an active vertex
and two segments emanating from it to the floor and ceiling curves,
respectively. We refer to such a double door as an \emphi{gate}.  As
such, specifying a single corridor requires encoding the vertices of
the two gates and the floor and ceiling curves, which can be done in
$O(1)$ space, given the two original curves. In particular, a corridor
is defined by four curves -- a ceiling curve, floor curve, and two
additional curves inducing the two gates.

Here is an informal argument why corridors indeed have the above
described structure. If there is any other curve, except the floor and
the ceiling, that interacts with the interior of the corridor, then
this curve would break the corridor into smaller sub-corridors, as
there would be a medial axis vertex corresponding to three curves in
the interior of the corridor, which is impossible.

\begin{figure}[t]
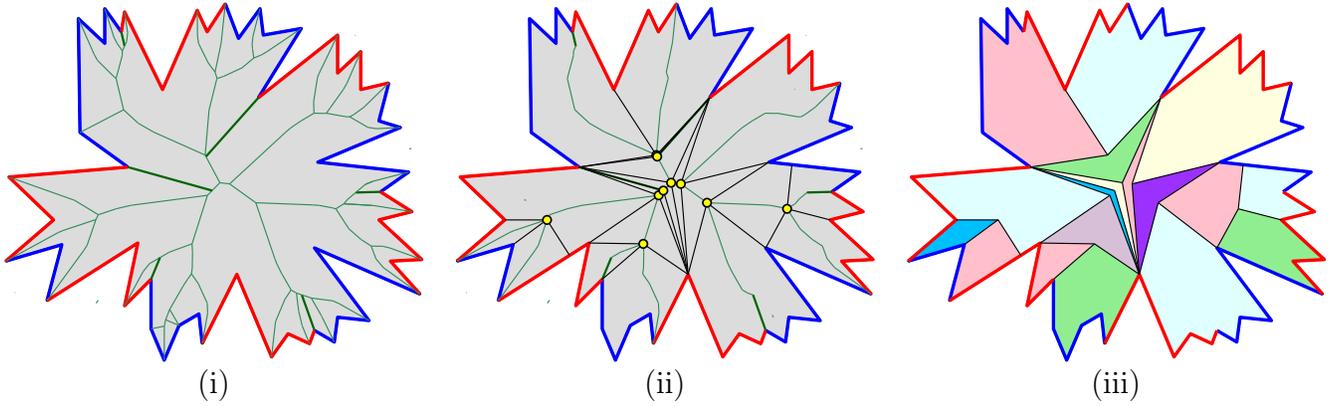

    \begin{center}
        \begin{tabular}{c%
               cc}
          \includegraphics[page=3,width=0.3\linewidth]%
          {\si{figs/corridor_decomp}}%
          &
            \includegraphics[page=5,,width=0.3\linewidth]%
            {\si{figs/corridor_decomp}}%
          &
            \includegraphics[page=7,,width=0.3\linewidth]%
            {\si{figs/corridor_decomp}}\\[-0.1cm]
          (i) & (ii) & (iii) 
        \end{tabular}
    \end{center}

    \vspace{-0.5cm}

    \caption{Another example of a corridor decomposition for a
       polygon: %
       (i) Input polygon and its curves and its medial axis (the thick
       lines are the angle bisectors for the obtuse angles where two
       curves meet), %
       (ii) active vertices and their spokes (with a reduced medial
       axis), and %
       (iii) the resulting corridor decomposition.%
    }
    \figlab{corridor:decomposition:polygon:2}
\end{figure}

\begin{figure}[t]
    \begin{center}
        \begin{tabular}{c%
               cc}
          \includegraphics[page=2,width=0.3\linewidth]%
          {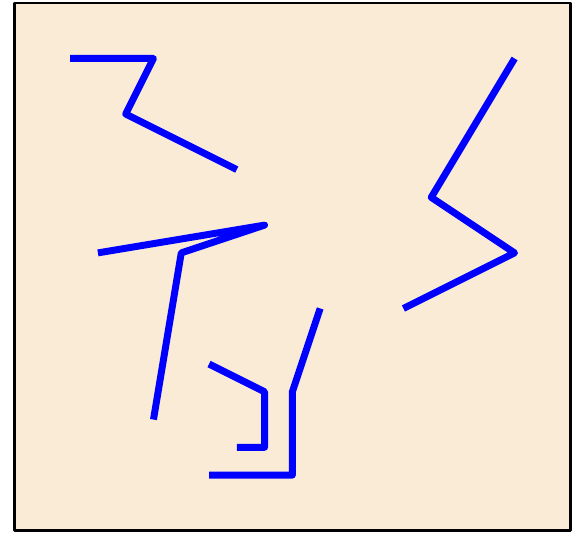}%
          &
            \includegraphics[page=4,width=0.3\linewidth]%
            {figs/medial_axis_curves}%
            &
              \includegraphics[page=5,,width=0.3\linewidth]%
              {figs/medial_axis_curves}\\[-0.1cm]
            (i) & (ii) & (iii) 
        \end{tabular}
    \end{center}

    \vspace{-0.5cm}

    \caption{Corridor decomposition for disjoint curves: %
       (i) Input curves, the medial axis, and active vertices, %
       (ii) the critical circles, and their spokes, and %
       (iii) the resulting corridor decomposition.%
    }
    \figlab{corridor:decomposition}
\end{figure}

\subsubsection{Corridor decomposition for a subset of the curves} %
For a subset $\CurveSetA \subseteq \CurveSet$, of total complexity
$\nB$, one can apply a similar construction. Again, compute the medial
axis of the curves of $\CurveSetA$, by computing, in
$O( \nB \log \nB)$ time, the Voronoi diagram of the segments used by
the curves \cite{f-savd-87}, and extracting the medial axis (it is now
a planar graph instead of a tree). Now, introducing the spokes of the
active vertices, results in a decomposition into corridors. For
technical reasons, it is convenient to add a large bounding box, and
restrict the construction to this domain, treating this \emphi{frame}
as yet another input curve.  \figref{corridor:decomposition} depicts
one such corridor decomposition.

Let $\CD{\CurveSetA}$ denote this resulting decomposition into
corridors.

\subsubsection{Properties of the resulting decomposition}

Every corridor in the resulting decomposition $\CD{\CurveSetA}$ is
defined by a constant number of input curves. Specifically, consider
the set of all possible corridors; that is
\begin{math}
    \AllCorridors = \bigcup_{\CurveSetB \subseteq \CurveSet}
    \CD{\CurveSetB}.
\end{math}
Next, consider any corridor $\Corridor \in \AllCorridors$, then there
is a unique \emphi{defining set}
$\DefSet{\Corridor} \subseteq \CurveSetA$ (of at most $4$
curves). Similarly, such a corridor has \emphi{stopping set} (or
\emphi{conflict list}) of $\Corridor$, denoted by
$\KillSet{\Corridor}$.  Here, the defining set of a corridor, as
depicted by \figref{doors}, is formed by the floor and the ceiling
curves of the corridor, and the two external curves defining the two
gates. The stopping set contains all the curves that intersect the
interior of the corridor, or the interior of the two disks that define
the two gates.

Consider any subset $\Sample \subseteq \CurveSet$.  It is easy to
verify that the following two conditions hold: %
\smallskip%
\begin{compactenum}[{\qquad\qquad}(i)]
    \item 
    For any $C \in \CD{\Sample}$, we have $\DefSet{C} \subseteq
    \Sample$ and $\Sample \cap \KillSet{C} = \emptyset$.
    
    \item 
    If $\DefSet{C} \subseteq \Sample$ and $\KillSet{ C }\cap
    \Sample=\emptyset$, then $C \in \CD{\Sample}$.
\end{compactenum}
\smallskip%
Namely, the corridor decomposition complies with the technique of
Clarkson-Shor (see \secref{p:l:violator:spaces} and
\defref{c:s:framework}).




\subsection{Computing a specific corridor}

Let $\pnt$ be a point in the plane, and let $\CurveSet$ be a set of
$\nA$ interior disjoint curves (stored in a read only memory), where
each curve is of complexity $\Space$. Let $n$ be the total complexity
of these curves (we assume that $n = \Theta(\Space \, \nA)$). Our
purpose here is to compute the corridor $\Corridor \in \CD{\CurveSet}$
that contains $\pnt$.  Formally, for a subset
$\CurveSetA \subseteq \CurveSet$, we define the function
$w(\CurveSetA)$, to be the defining set of the corridor
$\Corridor \in \CD{\CurveSetA}$ that contains $\pnt$. Note, that such
a defining set has cardinality at most $\cDim=4$.

\subsubsection{Basic operations} %
 We need to
specify how to implement the two basic operations:
\begin{compactenum}[\qquad(A)]
    \item (\textbf{Basis computation}) Given a set of $O(1)$ curves,
    we compute their medial axis, and extract the corridor containing
    $\pnt$. This takes $O( \Space \log \Space)$ time, and uses
    $O(\Space)$ work space.

    \item (\textbf{Violation test}) Given a corridor $\Corridor$, and
    a curve $\curve$, both of complexity $O(\Space)$, we can check if
    $\curve$ violates the corridor by checking if an arbitrary vertex
    of $\curve$ is contained in $\Corridor$ (this takes $O(\Space)$
    time to check), and then check in $O(\Space)$ time, if any segment
    of $\curve$ intersects the two gates of $\Corridor$. More
    precisely, each gate has an associated empty disk with it. The
    corridor is violated if the curve $\curve$ intersect these two
    empty disks associated with the two gates of the corridor,
    see \figref{doors}.  Overall, this takes $O(\Space)$ time.
\end{compactenum}

\smallskip%
\noindent%
We thus got the following result.

\begin{lemma}
    \lemlab{compute:corridor}%
    Given a polygon $\Polygon$ with $n$ vertices, stored in read only
    memory, and let $\Space$ be a parameter. Let $\CurveSet$ be the
    set of $\nA$ curves resulting from breaking $\Polygon$ into
    polygonal curves each with $\Space$ vertices, as described in
    \defref{breakup}. Then, given a query point $\pnt$ inside
    $\Polygon$, one can compute, in
    $O( n + \Space \log \Space \log^4 (n/\Space) )$ expected time, the
    corridor of $\CD{\CurveSet}$ that contains $\pnt$. This algorithm
    uses $O( \Space )$ work space.
\end{lemma}

\begin{proof}
    As described above, this problem is a point location problem in a
    canonical decomposition; that is, a violator space problem.
    Plugging it into \thmref{v:s:c:space}, makes in expectation
    $O(\nA)$ violation tests, each one takes $O(\Space)$ time, where
    $\nA = \floor{ (n -1) / \Space}+1 = O(n/\Space)$. The algorithm
    performs in expectation $O( \log^4 \nA)$ basis computations, each
    one takes $O(\Space \log \Space)$ time.  Thus, expected running
    time is $O( \nA \cdot \Space + \Space \log \Space \log^4 \nA ) $
    time.  As for the space required, observe that a corridor can be
    described using $O(1)$ space. 

    The algorithm needs $O(\Space)$ work space to implement the basis
    computation operation, as all other portions of the algorithm
    require only constant additional work space.
\end{proof}%

We comment that the problem of \lemref{compute:corridor} can also be
solved as an \LP-type problem. We refer the interested reader to
\apndref{l:p:type}.%

\section{Shortest path in a polygon in sublinear %
   space}
\seclab{shortest:path}

Let $\Polygon$ be a simple polygon with $n$ edges in the plane, and
let $\source$ and $\target$ be two points in $\Polygon$, where $s$ is
the \emphi{source}, and $t$ is the \emphi{target}.  Our purpose here
is to compute the shortest path between $s$ and $t$ inside $\Polygon$.
The vertices of $\Polygon$ are stored in (say) counterclockwise order
in an array stored in a read only memory.  Let $\Space$ be a
prespecified parameter that is (roughly) the amount of additional
space available for the algorithm.

\begin{figure}[t]
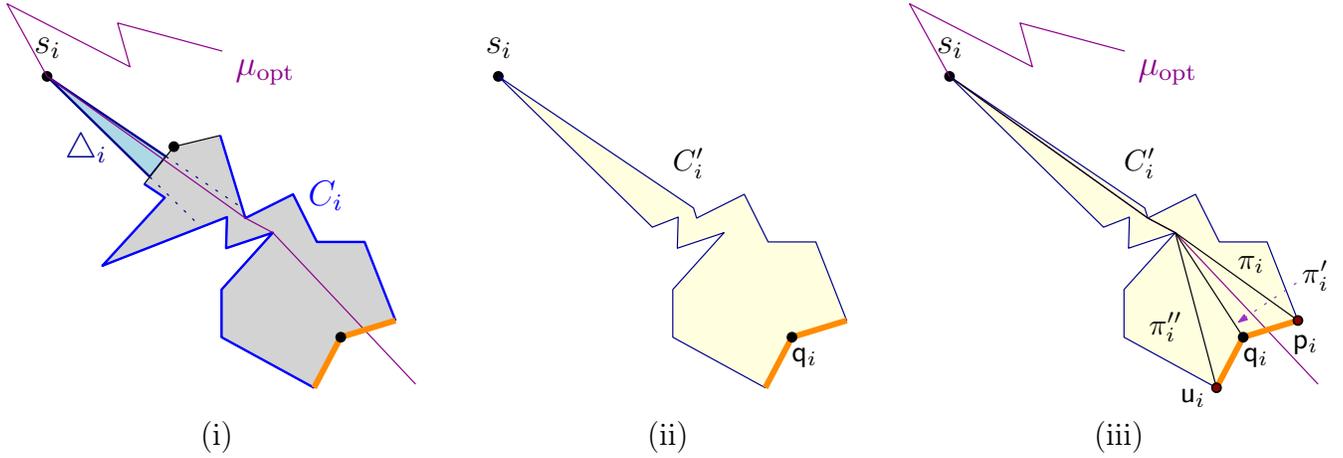

    \centerline{%
       \begin{tabular}{cc%
              c}
           \includegraphics[page=1,width=0.3\linewidth]%
         {\si{figs/ith_iteration}}%
         &%
           \includegraphics[page=2,width=0.3\linewidth]%
           {\si{figs/ith_iteration}}%
           &%
           \includegraphics[page=3,width=0.3\linewidth]%
             {\si{figs/ith_iteration}}%
         \\%
         (i) & (ii) & (iii)
       \end{tabular}%
    }
    \caption{(i) The state in the beginning of the $i$\th iteration. %
       (ii) The clipped polygon $\Corridor_i'$. %
       (iii) The funnel created by the shortest paths from $\source_i$
       to the two spoke endpoints.%
    }
    \figlab{clip:clap}
\end{figure}

\subsection{Updating the shortest path through a corridor}

We remind the reader that a corridor has two double \emph{doors}, see
\figrefpage{doors}. The rest of the boundary of the corridor is made
out two chains from the original polygon.

Given two rays $\ray$ and $\rayA$, that share their source vertex
$\vertex$ (which lies inside $\Polygon$), consider the polygon
$\PolygonA$ that starts at $v$, follows the ray $\ray$ till it hits
the boundary of $\Polygon$, then trace the boundary of $\Polygon$ in a
counterclockwise direction till the first intersection of $\rayA$ with
$\partial \Polygon$, and then back to $\vertex$. The polygon
$\PolygonA = \ClipPoly{\Polygon}{\ray}{\rayA}$ is the \emphi{clipped}
polygon.  See \figref{clip:clap} (ii).

A \emphi{geodesic} is the shortest path between two points (restricted
to lie inside $\Polygon$). Two geodesics might have a common
intersection, but they can cross at most once. Locally, inside a
polygon, a geodesic is a straight segment.  For our algorithm, we need
the following two basic operations (both can be implemented using
$O(1)$ work space): %
\smallskip%
\begin{compactenum}[\quad(A)]
    \item \AlgIsIn{}($\pnt$): Given a query point $\pnt$, it decides
    if $\pnt$ is inside $\Polygon$. This is done by scanning the edges
    of $\Polygon$ one by one, and counting how many edges cross the
    vertical ray shooting from $\pnt$ downward. This operation takes
    linear time (in the number of vertices of $\Polygon$). If the
    count is odd then $\pnt$ is inside $\Polygon$, and it is outside
    otherwise\footnote{This is the definition of inside/outside of a
       polygon by the Jordan curve theorem.}.

    \item \AlgIsInSubPoly{}($\pnt, \ray, \rayA$): returns true if
    $\pnt$ is in the clipped polygon
    $\ClipPoly{\Polygon}{\ray}{\rayA}$.  This can be implemented to
    work in linear time and constant space, by straightforward
    modification of \AlgIsIn{}.

    \SaveEnumCounter{}
\end{compactenum}
\medskip

Using vertical and horizontal rays shot from $\source$, one can
decide, in $O(n)$ time, which quadrant around $\source$ is locally
used by the shortest path from $\source$ to $\target$. Assume that
this path is in the positive quadrant. It would be useful to think
about geodesics starting at $\source$ as being sorted
angularly. Specifically, if $\Path$ and $\PathA$ are two geodesic
starting at $\source$, then $\Path$ is to the \emphi{left} of
$\PathA$, if the first edge of
$\Path$ is counterclockwise to the first edge of $\PathA$. If the
prefix of $\Path$ and $\PathA$ is non-empty, we apply the same test to
the last common point of the two paths. Let $\Path \prec \PathA$
denote that $\Path$ is to the left of $\PathA$.

 In particular, if the endpoint of the rays $\ray, \rayA$ is the
source vertex $\source$, and the geodesic between $\source$ and
$\target$ lies in $\ClipPoly{\Polygon}{\ray}{\rayA}$, then given a
third ray $\rayB$ lying between $\ray$ and $\rayA$, the shortest path
between $\source$ and $\target$ in $\Polygon$ must lie completely
either in $\ClipPoly{\Polygon}{\ray}{\rayB}$ or
$\ClipPoly{\Polygon}{\rayB}{\rayA}$, and this can be tested by a
single call to \AlgIsInSubPoly{} for checking if $\target$ is in
$\ClipPoly{\Polygon}{\rayB}{\rayA}$.


\subsubsection{Limiting the search space}

\begin{lemma}%
    \lemlab{c:c:w}%
    Let $\Polygon$, $\source$ and $\target$ be as above, and $\SPath$
    be the shortest path from $\source$ to $\target$ in $\Polygon$.
    Let $\pnt \pntA$ be a given edge, which is the last edge in the
    shortest path $\Path$ from $\source$ to $\pntA$, where $\pntA$ is
    in $\Polygon$.  Then, one can decide in $O(n)$ time, and using
    $O(1)$ space, if $\SPath \prec \Path$, where $n$ is the number of
    vertices of ${\Polygon}$.
\end{lemma}

\begin{proof}
    If $\pnt = \source$, then this can be determined by shooting two
    rays, one in the direction of $\source \rightarrow \pntA$ and the
    other in the opposite direction, where
    $ \source \rightarrow \pntA$ denotes the vector $\pntA - \source$.
    Then a single call to \AlgIsInSubPoly{} resolves the issue.

    Otherwise, $\pnt$ must be a vertex of $\Polygon$. Consider the ray
    $\ray$ starting at $\pnt$ (or $\pntA$ -- that does not matter) in
    the direction of $\pnt \rightarrow \pntA$. Compute, in linear
    time, the first intersection of this ray with the boundary of
    $\Polygon$, and let $\pntB$ be this point. Clearly, $\pnt \pntB$
    connects two points on the boundary of $\Polygon$, and it splits
    $\Polygon$ into two polygons. One can now determine, in linear
    time, which of these two polygons contains $\target$, thus
    resolving the issue.
\end{proof}%

\subsubsection{Walking through a corridor}

In the beginning of the $i$\th iteration of the algorithm it would
maintains the following quantities (depicted in \figref{clip:clap}
(i)):
\begin{compactenum}[\quad\quad(A)]
    \item $\source_i$: the current source (it lies on the shortest
    path $\sopt$ between $\source$ and $\target$).
    
    \item $\Corridor_i$: The current corridor.
    \item $\triangle_i$: A triangle having $\source_i$ as one of its
    vertices, and its two other vertices lie on a spoke of
    $\Corridor_i$. The shortest path $\sopt$ passes through
    $\source_i$, and enters $\Corridor_i$ through the base of
    $\triangle_i$, and then exits the corridor through one of its
    ``exit'' spokes.
\end{compactenum}

\smallskip The task at hand is to trace the shortest path through
$\Corridor_i$, in order to compute where the shortest path leaves the
corridor.

\begin{lemma}
    \lemlab{step}%
    Tracing the shortest path $\sopt$ through a single corridor takes
    $O\bigl( n \log \Space + \Space \log \Space \log^4 \nA \bigr)$
    expected time, using $O(\Space)$ space.
\end{lemma}

\begin{proof}
    We use the above notation.  The algorithm glues together
    $\triangle_i$ to $\Corridor_i$ to get a new polygon. Next, it
    clips the new polygon by extending the two edges of $\triangle_i$
    from $\source_i$. Let $\Corridor_i'$ denote the resulting polygon,
    depicted in \figref{clip:clap} (ii).  Let the three vertices of
    $\Corridor_i'$ forming the two ``exit'' spokes be
    $\pnt_i,\pntA_i, \pntB_i$.  Next, the algorithm computes the
    shortest path from $\source_i$ to the three vertices
    $\pnt_i, \pntA_i, \pntB_i$ inside $\Corridor_i'$, and let
    $\curveA_i, \curveA_i', \curveA_i''$ be these paths, respectively
    (this takes $O( \cardin{\Corridor_i'}) = O( \Space)$ time
    \cite{gh-ospqs-89} after linear time triangulation of the polygon
    \cite{c-tsplt-91a, agr-ratsp-01}). Using \lemref{c:c:w} the
    algorithm decides if $\curveA_i \prec \sopt \prec \curveA_i'$ or
    $\curveA_i' \prec \sopt \prec \curveA_i''$.  We refer to a prefix
    path (that is part of the desired shortest path) followed by the
    two concave chains as a \emphi{funnel} -- see \figref{clip:clap}
    (iii) and \figref{funnel:reduction} for an example.

        Assume that $\curveA_i \prec \sopt \prec \curveA_i'$, and let
        $F_i$ be the funnel created by these two shortest paths, where
        $\pnt_i\pntA_i$ is the base of the funnel. If the space
        bounded by the funnel is a triangle, then the algorithm sets
        its top vertex as $\source_{i+1}$, the funnel triangle is
        $\triangle_{i+1}$, and the algorithm computes the corridor on
        the other side of $\pnt_i \pntA_{i}$ using the algorithm of
        \lemref{compute:corridor} (by picking any point on
        $\pnt_i \pntA_i$ as the query for the point location
        operation), set it as $\Corridor_{i+1}$, and continues the
        execution of the algorithm to the next iteration.

        \SaveIndent
    \noindent%
    \begin{minipage}{0.78\linewidth}%
        \RestoreIndent%
        The challenge is handling funnel chains that are
        ``complicated'' concave polygons (with at most $O(\Space)$
        vertices), see \figref{funnel:reduction}. As long as the
        funnel $F_i$ is not a triangle, pick a middle edge $\edge$ on
        one side of the funnel, and extend it till it hits the edge
        $\pnt_i \pnt_{i+1}$, at a point $\pntC$. This breaks $F_i$
        into two funnels, and using the algorithm of \lemref{c:c:w} on
        the edge $\edge$, decide which of these two funnels contains
        the shortest path $\sopt$, and replace $F_i$ by this
        funnel. Repeat this process till $F_i$ becomes a
        triangle. Once this happens, the algorithm continues to the
        next iteration as described above. Clearly, this funnel
        ``reduction'' requires $O( \log \Space )$ calls to the
        algorithm of \lemref{c:c:w}.

        Note, that the algorithm ``forgets'' the portion of the funnel
        that is common to both paths as it moves from $\Corridor_i$ to
        $\Corridor_{i+1}$. This polygonal path is a part of the
        shortest path $\sopt$ computed by the algorithm, and it can be
        output at this stage, before moving to the next corridor
        $\Corridor_{i+1}$.
    \end{minipage}
    \hfill
    \begin{minipage}{0.18\linewidth}%
        \includegraphics[width=0.95\linewidth]{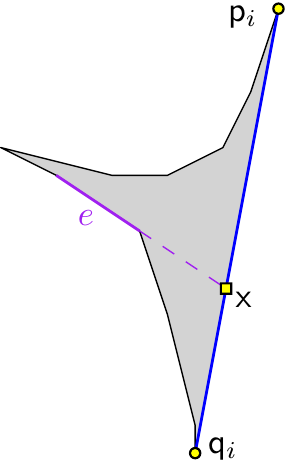}%
        \captionof{figure}{\\
           Funnel reduction.}  \figlab{funnel:reduction}
    \end{minipage}
    \smallskip%

    In the end of the iteration, this algorithm computes the next
    corridor $\Corridor_{i+1}$ by calling the algorithm of
    \lemref{compute:corridor}.
\end{proof}


\subsection{The algorithm}

The overall algorithm works by first computing the corridor
$\Corridor_1$ containing the source $\source_1=\source$ using
\lemref{compute:corridor}. The algorithm now iteratively applies
\lemref{step} till arriving to the corridor containing $\target$,
where the remaining shortest path can be readily computed. Since every
corridor gets visited only once by this walk, we get the following
result.

\begin{theorem}
    Given a simple polygon $\Polygon$ with $n$ vertices (stored in a
    read only memory), a start vertex $\source$, a target vertex
    $\target$, and a space parameter $\Space$, one can compute the
    length of the shortest path from $\source$ to $\target$ (and
    output it), using $O(\Space)$ additional space, in
    $O \pth{ {n^2}/{\Space}}$ expected time, if
    $\Space = O(n /\log^2 n)$. Otherwise, the expected running time is
    $O\pth{ {n^2}/{\Space} + n \log \Space \log^4 (n/\Space) }$.
\end{theorem}
\begin{proof}
    The algorithm is described above, and let
    $\nA = \floor{ (n -1) / \Space}+1$.  There are $O(\nA)$ corridors,
    and this bounds the number of iterations of the algorithm. As
    such, the overall expected running time is
    \begin{math}
        O\pth{ \nA \pth{n \log \Space + \Space \log \Space \log^4 \nA
           } }%
        =%
        O\pth{ ({n^2}/{\Space}) \log \Space + n \log \Space \log^4 \nA
        }.
    \end{math}

    To get a better running time, observe that the extra log factor
    (on the first term), is rising out of the funnel reduction
    $O( \log \Space)$ queries inside each corridor, done in the
    algorithm of \lemref{step}. If instead of reducing a funnel all
    the way to constant size, we reduce it to have say, at most
    $\ceil{\Space/4}$ edges (triggered by the event that the funnel
    has at least $\Space/2$ edges), then at each invocation of
    \lemref{step}, only a constant number of such queries would be
    performed. One has to adapt the algorithm such that instead of a
    triangle entering a new corridor, it is a funnel. The adaptation
    is straightforward, and we omit the easy details. The improved
    running time is
    \begin{math}
        O\pth{ {n^2}/{\Space} + n \log \Space \log^4 \nA }.%
    \end{math}
\end{proof}


\section{Epilogue}

In this paper, we showed the following: \smallskip%
\smallskip%
\begin{compactenum}[(A)]
    \item Violator spaces are the natural way to phrase the problem of
    one-shot point location query in a canonical decomposition.

    \item A variant of Seidel's algorithm solves Violator spaces in
    expected linear time.

    \item This algorithm can be modified to use only $O(1)$ work
    space, by using pseudo-random sequences.

    \item Backward analysis works if one uses pseudo-random sequences
    instead of true random permutations (\lemref{backward:a:works}).

    \item Last, but not least, shortest path in a polygon can be
    computed in $O(n^2 / \Space)$ expected time, using $O(\Space)$
    work space. This is done by breaking a polygon into $O(n/\Space)$
    subpolygons of size $O(\Space)$, and walking through these
    corridors, spending $O(n)$ time in each corridor.
\end{compactenum}
\smallskip%
While deriving some of the above results required quite a bit of
technical mess\footnote{The author would like to use this historical
   opportunity to apologize both to the reader and to himself for this
   mess.} %
(specifically (C) and (D) above), conceptually the basic ideas
themselves are simple and natural.

\medskip

The most interesting open problem remaining from our work, is whether
one can improve the running time for computing the shortest path in a
polygon with $O(\Space)$ space to be faster than $O(n^2/\Space)$.


\subsection*{Acknowledgments}

The author became aware of the low-space shortest path problem during
Tetsuo Asano talk in the Workshop in honor of his 65\th birthday
during SoCG 2014. The author thanks him for the talk, and the
subsequent discussions. The author also thanks Pankaj Agarwal, Chandra
Chekuri, Jeff Erickson and Bernd \si{G\"a\si{rtner}} for useful
discussions. The author also thanks the anonymous referees, for both
the conference and journal versions, for their detailed comments.


\InJoCGVer{%
   \BibTexMode{
      \bibliographystyle{plain}%
      \bibliography{sublinear_space}%
   }%
   \BibLatexMode{
      \printbibliography
   }
}%
\InNotJoCGVer{%
   \BibTexMode{%
      \bibliographystyle{salpha}%
      \bibliography{sublinear_space}%
   }
   \BibLatexMode{
      \printbibliography
   }
}


\appendix

\section{Backward analysis for pseudo-random %
   sequences}
\apndlab{backward}

We prove the probabilistic bounds we need from scratch. We emphasize
that this is done so that the presentation is self contained. Our
estimates and arguments are inspired by Mulmuley's \cite[Chapter
10]{m-cgitr-94}, although we are dealing with somewhat different
events. In particular, there is nothing new in \apndref{s:p:s}, and
relatively little new in \apndref{p:s:b:c}.  Finally, \apndref{b:b:a}
proves the new bounds we need.

Interestingly, Indyk \cite{i-samwi-01} used arguments in a similar
spirit to bound a different event, related to the probability of the
$i$\th element in a pseudo-random ``permutation'' to be a minimum of
all the elements seen so far. This does not seem to have a direct
connection to the analysis here.

\subsection{Some probability stuff}
\apndlab{s:p:s}

\newcommand{\dInd}{\nu}

We need the following lemma \cite[Theorem A.2.1]{m-cgitr-94} and
provide a somewhat more precise bounds on the constants involved, than
the reference mentioned.

\bigskip%
\noindent\textbf{Restatement of \lemref{chebychev:k}.~}%
{\em \LemmaCehbychevKBody}

\begin{proof}
    For $i > 1$, and any $j$, let
    \begin{align*}
      \alpha_i%
      &=%
        \Ex{ \Bigl. (X_j - p)^i}%
        =%
        p (1-p)^i + (1-p) (-p)^i%
        \leq %
        p (1-p)\pth{ \Bigl. (1-p)^{i-1} + p^{i-1}}%
        \InJoCGVer{%
      \\&}%
        \leq%
        p (1-p)%
          \leq%
          p.
    \end{align*}
    Similarly, we have that
    $\alpha_1 = \Ex{ X_j - p } = \Ex{X_j} - p = p - p = 0$.

    \smallskip%
    Consider the variable
    $Z = \pth{\bigl. \sum_{j=1}^n (X_j -p)}^{2\pDimB}$, and its
    expectation.  Let $\ISet$ be the set of all $2\pDimB$ tuples
    $\pth{i_1, \ldots, i_{2\pDimB}} \in \IntRange{n}^{2\pDimB}$, such
    that $i_1 \leq i_2 \leq \ldots \leq i_{2\pDimB}$. We have by
    linearity of expectations that
    \begin{math}
        \Ex{\bigl.Z\ts} = \sum_{\pth{i_1, \ldots, i_{2\pDimB}} \in
           \ISet}
        \Ex{\prod_{j=1}^{2\pDimB} \pth{X_{i_j} - p}}.
    \end{math}
    Now, any tuple with $(i_1,\ldots, i_{2\pDimB})$ such that an
    index, say $i_\ell$, appears exactly once, contributes $0$ to this
    summation, since then
    \begin{align*}
      \Ex{\prod\nolimits_{j=1}^{2\pDimB} \pth{X_{i_j} - p}}%
      = %
      \Ex{\Bigl. X_{i_\ell} - p } \Ex{\prod\nolimits_{j\in
      \IntRange{2\pDimB}\setminus \brc{\ell}} \pth{X_{i_j} - p}} =
      0,
    \end{align*}
    as $\alpha_1 = 0$, and as the variables are $2\pDimB$-wise
    independent. So, consider a tuple $I \in \ISet$, where every index
    appears at least twice, and there are $\dInd \leq \pDimB$ distinct
    indices.  In particular, if such a term involves variables
    $i_1,\ldots, i_\dInd$, with multiplicities $n_1, \ldots, n_\dInd$
    (all at least $2$), then
    \begin{align*}
        \Ex{\Bigl.{\prod\nolimits_{j=1}^\dInd \pth{X_{i_j} -
                 p}^{n_j}}}%
        = %
        \prod\nolimits_{j=1}^\dInd \Ex{\Bigl.\pth{X_{i_j} - p}^{n_j}}%
        \leq \prod\nolimits_{j=1}^\dInd \alpha_{n_j}%
        \leq%
        p^\dInd,
    \end{align*}
    using the $2\pDimB$-wise independence.

    We want to bound the total sum of the coefficients of all such
    monomials in the polynomial $Z$ that have at most $\dInd$ distinct
    variables.  To this end, to specify how we generated such a
    monomial, we need to specify $2\pDimB$ integers
    $i_1, \ldots, i_{2\pDimB}$.  We first scan such a sequence and
    write down the $\dInd$ unique values encountered, in the order
    they are being encountered. There are $\leq n^\dInd$ such ``seen
    first'' sequences. Now, for every index $i_t$, we need to specify
    if it is new or not. There are $\binom{2\pDimB}{\dInd}$ such
    choices. Finally, for each number in this sequence which is a
    repetition, we need to specify which of the previously (at most
    $\dInd$ numbers) seen before it is. There are
    $\dInd^{2\pDimB-\dInd}$ such possibilities. As such, the total sum
    of all such coefficients in $Z$ is bounded by
    $\binom{2\pDimB}{\nu}n^{\dInd}\dInd^{2\pDimB-\dInd}$. Every such
    monomial contributes $p^\dInd$ to $\Ex{Z}$. As such, we have
    \begin{align*}
      \Ex{\Bigl.Z}%
      &\leq%
        0 + \sum_{\dInd=1}^{\pDimB}
        \binom{2\pDimB}{\nu}n^{\dInd}\dInd^{2\pDimB-\dInd} p^{\dInd}
        \leq%
        \pth{pn}^\pDimB \sum_{\dInd=1}^{\pDimB} \binom{2\pDimB}{\nu}\pDimB^{2\pDimB-\dInd}
        \leq%
        \frac{2^{2\pDimB}}{\sqrt{2\pDimB}}%
        \pDimB^{2\pDimB} \pth{pn}^\pDimB%
        =%
        \frac{ \pth{4\pDimB^2 pn}^\pDimB}{\sqrt{2\pDimB}},
    \end{align*}
    since $p n \geq 1$, and
    \begin{math}
        \binom{2\pDimB}{\dInd} \leq \binom{2\pDimB}{\pDimB} \leq
        2^\pDimB/\sqrt{2\pDimB}
    \end{math}
    (see \cite[Section 2.6]{mn-idm-98}).

    By Markov's inequality, we have that
    \begin{align*}
        \ds%
        \Prob{\Bigl. \cardin{Y - \mu} \geq \mu}%
        =%
        \Prob{ \pth{Y - \mu}^{2\pDimB} \geq \mu^{2\pDimB}}%
        \leq%
        \frac{\Ex{ \pth{Y - \mu}^{2\pDimB}} } {\mu^{2\pDimB}}%
        =%
        \frac{\bigl.\Ex{\bigl.Z }}{\mu^{2\pDimB}}%
        \leq%
        \frac{ \pth{4\pDimB^2}^\pDimB \mu^\pDimB}{ \sqrt{2\pDimB} \mu^{2\pDimB}},
    \end{align*}
    implying the claim.
\end{proof}




\subsection{The probability of a specific basis to %
   have a conflict at iteration $i$}
\apndlab{p:s:b:c}

The following lemma bounds the probability that a specific basis
$\Basis$ is defined in the prefix of the $i-1$ sampled elements, and
the first element violating it is in position $i$. This probability is
of course affected by the size of its conflict list $\CLList$.

\begin{lemma}
    \lemlab{specific:basis}%
    Let $X_1, \ldots X_i$ be a sequence of variables that are
    uniformly distributed over $\IntRange{n}$, that are $\pDim$-wise
    independent, where $\pDim > \cDim$ is a sufficiently large
    constant. Let $\Basis, \CLList \subseteq \IntRange{n}$ be two
    disjoint sets of size $\cDim$ and $k = t(n/i)$, respectively,
    where $\cDim$ is a constant, $i >2\cDim$, and $t>0$ is
    arbitrary. Let $\Event$ be the following event:

    \smallskip
    \begin{inparaenum}[\quad(I)]
        \item $\Basis \subseteq \brc{X_1,\ldots, X_{i-1}}$.
        \item $\CLList \cap  \brc{X_1,\ldots, X_{i-1}} = \emptyset$.
        \item $X_i \in \CLList$.
    \end{inparaenum}
    
    \smallskip%
    Then, for $t > 1$, we have
    $ \Prob{\bigl. \Event} = O\pth{ \frac{1}{i}
       \pth{\frac{i}{n}}^\cDim {\frac{1}{t^{\pDimB}}} }$,
    where $\pDimB = (\pDim-\cDim-1)/2 -1$.

    For $t \leq 1$, we have
    $O\pth{ \frac{1}{i} \pth{\frac{i}{n}}^\cDim }$, as long as
    $\pDim \geq \cDim$.
\end{lemma}

\begin{proof}
    Fix the numbering of the elements of $\Basis$ as
    $\Basis = \brc{b_1, \ldots, b_\cDim}$, and let $\ISet$ be the set
    of tuples of $\cDim$ distinct indices in
    $\IntRange{i-1}^{\cDim}$. For such a tuple
    $I =(i_1,\ldots, i_\cDim)$, let $\Event(I)$ be the event that
    $X_{i_1} = b_1, \ldots, X_{i_\cDim}= b_\cDim , \text{ and } X_{i}
    \in \CLList$. We have that
    \begin{align*}
        \Prob{ \Bigl. \Event(I) }%
        &=%
        \Prob{\Bigl. X_i \in \CLList}%
        \prod_{j=1}^{\cDim} \Prob{\Bigl. X_{i_j} = b_j}%
        =%
        \frac{1}{n^{\cDim} } \cdot \frac{k}{n} =
        \frac{k}{n^{\cDim+1}},
    \end{align*}
    as the variables $X_1, \ldots, X_n$ are $\pDim$-wise independent,
    and $\pDim > \cDim+1$.  Let $\EventA(I)$ be that event that none
    of the variables of $\XS = \brc{X_1,\ldots, X_{i-1}}$ are in
    $\CLList$. In particular, consider the $m = i-1-\cDim$ variables
    in $\XS$ that have an index in $\IntRange{i-1} \setminus I$, and
    let $X_j'$ denote the $j$\th such variable. Let $Y_j$ be an
    indicator variable that is one if $X_j' \in \CLList$. The
    variables $X_1', \ldots, X_m'$ are (at worst)
    $(\pDim - \cDim-1)$-wise independent (since, conceptually, we
    fixed the value of the variables of $I$), and so are the indicator
    variables $Y_1,\ldots, Y_m$.  For any $j$, we have that
    $p = \Ex{Y_j} = \Prob{X_j' \in \CLList} = k/n$.  As such, for
    $\pDimB = (\pDim-\cDim-1)/2$, we have that
    \begin{align*}
      \Prob{\Bigl. \EventA(I) \cap \Event(I) }%
      &=%
        \Prob{\EventA(I) \sep{\Bigl. \Event(I)}}
        \Prob{\Bigl. \Event(I)}%
        =%
        \Prob{\sum\nolimits_j Y_j = 0} \frac{k}{n^{\cDim+1}}%
      \\&%
          =%
          O\pth{ \pth{\frac{1}{pm}}^{(\pDim - \cDim -1)/2}
          \frac{k}{n^{\cDim+1}}}%
          =%
          O\pth{ \pth{\frac{n}{k(i-1-\cDim)}}^{\pDimB}
          \frac{k}{n^{\cDim+1}}}\\
      &=%
        O\pth{ \pth{\frac{n}{ki}}^{\pDimB} \frac{(tn/i)}{n^{\cDim+1}}}%
        =%
        O\pth{ {\frac{1}{t^{\pDimB}}} \cdot
        \frac{t}{ i n^{\cDim}}}
        =%
        O\pth{ {\frac{1}{i t^{\pDimB-1} n^{\cDim}}}},
    \end{align*}
    by \lemref{chebychev:k}, as $i > 2\cDim$, $\pDim$ is constant, and
    $k = tn/i$.  There are at most $i^\cDim$ choices for the tuple
    $I$, and there are $\cDim!$ different orderings of the elements of
    $\Basis$. As such, the desired probability is bounded by
    \begin{math}
        \ds \Biggl.O\pth{ { i^\cDim \cDim!  {\frac{1}{i t^{\pDimB-1}
                    n^{\cDim} } } }}%
        =%
        O\biggl( \frac{1}{i} \Bigl(\frac{i}{n} \Bigr)^\cDim
        {\frac{1}{t^{\pDimB-1}}} \biggr),
    \end{math}
    as $\cDim$ is a constant.

    The bound for $t\leq 1$, follows by using $\Event(I)$ instead of
    $\EventA(I) \cap \Event(I)$ in the above analysis.
\end{proof}

\subsection{Back to backward analysis}
\apndlab{b:b:a}

Consider a fixed violator space $\VSpace = (\CSet, \vFuncC)$ with
$n =\cardin{\CSet}$ constraints, and combinatorial dimension
$\cDim$. Let $\BasesAll$ be the set of bases of $\VSpace$.  Let
$\XS = X_1, \ldots, X_n$ be the sequence of constraints generated by
the pseudo-random generator, that is $\pDim$-wise independent (say,
somewhat arbitrarily, $\pDim = 6\cDim + 10$).  In the following, the
\emph{$i$\th prefix} of $\XS$ is $\XS_i = \permut{ X_1,\ldots, X_i}$.

\begin{lemma}
    \lemlab{count:bases}%
    Let $\BasisLE{k}$ be the set of all the bases of $\VSpace$ that
    have a conflict list of size at most $k$. Then, we have that
    $\cardin{\BasisLE{k}} =O(k^\cDim)$.
\end{lemma}
\begin{proof}
    Follows readily from the Clarkson-Shor \cite{c-arscg-88,
       cs-arscg-89} technique. Pick every constraint into $\RSample$
    with probability $1/k$ into a sample $\RSample$. The probability
    of a basis with conflict list of size at most $k$ to survive this
    purge and be the basis of $\RSample$ is $\Omega(1/k^\cDim)$. Since
    $\basisX{\RSample}$ is a single basis, the result readily follows.
\end{proof}


\begin{lemma}
    \lemlab{p:light}%
    Assume that $\pDim \geq \cDim$.  For $i > 2\cDim$, let
    $\probEL{i}$ be the probability that $X_i$ violates
    $\Basis = \basisX{X_1, \ldots, X_{i-1}} \in \BasisLE{k}$, where
    $k= n/i$. Then $\probEL{i} = O( 1/i)$.
\end{lemma}
\begin{proof}
    For a fixed basis $\Basis \in \BasisLE{k}$ the desired
    probability, for $\Basis$ being the basis, is
    \begin{math}
        \probEL{i, \Basis} = O\pth{ \frac{1}{i}
           \pth{\frac{i}{n}}^\cDim },
    \end{math}
    by \lemref{specific:basis} (for $t \leq 1$).  By
    \lemref{count:bases}, we have
    $\cardin{\BasisLE{k}} = O( (n/i)^\cDim)$. As such, the desired
    probability is bounded by
    \begin{math}
        O\pth{ ({1}/{i}) \pth{{i}/{n}}^\cDim (n/i)^\cDim} = O(1/i).
    \end{math}
\end{proof}

A basis $\Basis$ is \emphi{$t$-heavy} for $\XS_i$, if
$\ds \frac{\cardin{\clX{\Basis}}}{n/i} \in (t-1, t]$.

\begin{lemma}%
    \lemlab{p:heavy}%
    Assume that $\pDim$ is a constant and $\pDim \geq 3(\cDim+3)$.
    For $i > 2\cDim$, let $\probEH{t}{i}$ be the probability that
    $X_i$ violates
    \begin{math}
        \Basis = \basisX{X_1, \ldots, X_{i-1}} \in \BasesAll,
    \end{math}
    and $\Basis$ is $t$-heavy, for $t > 2$.  Then
    $\probEH{t}{i} = O\pth{ 1/(i t^2)}$.
\end{lemma}
\begin{proof}
    Arguing as in \lemref{p:light}, there are
    $O\bigl( \pth{{tn}/{i}}^\cDim \bigr)$ $t$-heavy bases, and by
    \lemref{specific:basis}, the probability of the desired event for
    a specific basis $\Basis'$, with
    $\cardin{\clX{\Basis'}} \geq (t-1)n/i$, is bounded by
    $O\pth{ \frac{1}{i} \pth{\frac{i}{n}}^\cDim
       {\frac{1}{t^{\pDimB}}}}$.
    As such, the desired probability is bounded by
    \begin{math}
        \ds O\pth{\pth{\frac{tn}{i}}^\cDim \frac{1}{i}
           \pth{\frac{i}{n}}^\cDim \frac{1}{t^{\pDimB}} }%
        =%
        O \pth{ \frac{1}{i t^2}},
    \end{math}
    if $\pDimB \geq \cDim + 2$, which holds for
    $\pDim \geq 3(\cDim+3)$.
\end{proof}

\noindent\textbf{Restatement of \lemref{backward:a:works}. }%
 {\em \LemmaBAnalysisBody}

\begin{proof}
    Plugging in the bounds of \lemref{p:light} and \lemref{p:heavy},
    we have that the desired probability is bounded by
    $\probEL{i} + \sum_{t=2}^\infty \probEH{t}{i} = O\pth{1/i +
       \sum_{t=2}^\infty \frac{1}{i t^2}} =O(1/i)$, as claimed.
\end{proof}


\section{point location in corridors using \LP-type %
   solver}
\apndlab{l:p:type}

For the amusement of the interested reader, we briefly describe how
point location in corridor decomposition can be solved as an \LP-type
problem. This requires only defining a ``crazy'' ordering on the
corridors that contain our query point, and the rest follows
readily. Of course, it is more elegant and natural to do this using
violator spaces.

We need to define an explicit ordering on the corridors so we can use
the \LP-type algorithm. This is somewhat tedious.  The corridor
$\Corridor$ corresponds to an edge $\edge$ of the medial axis
connecting the two active vertices $\pntB$ and $\pntC$ of the medial
axis.  Let $\DefSetChar$ be the set of four inputs curves that define
the corridor $\Corridor$ (i.e., $\DefSetChar$ is the \emph{defining
   set} of $\Corridor$).  The closest curve to the query point $\pnt$
is the \emphi{floor} of $\Corridor$. The other curve of $\DefSetChar$
that lies on the other side of the medial axis edge $\edge$, is the
\emphi{ceiling} of the corridor. Consider the closest point $\pnt'$ on
the medial axis to $\pnt$, and consider the two paths from $\pnt'$ to
$\pntB$ and $\pntC$. These two paths diverge at a point $\pntD$. The
curve in $\DefSetChar$ corresponding to the active vertex for which
the path diverges to the right (resp. left) at $\pntD$, is the
\emphi{right} (resp. \emphi{left}) curve of $\Corridor$.

We can now define an order on two corridors $w(F), w(F')$ that contain
$\pnt$. Formally, we define 
\begin{compactitem}
    \item If the floor of $w(F)$ is closer to $\pnt$ than $w(F')$,
    then we define $w(F) \succ w(F')$.

    \item Otherwise, by our general position assumption, it must be
    that the floors of $w(F)$ and $w(F')$ are the same.  If the
    ceilings of $w(F)$ and $w(F')$ are different, then consider the
    ceiling $c$ of $w(F \cup F')$. If $c$ is also the ceiling of
    $w(F)$, then $w(F) \succ w(F')$, otherwise, $c$ is the ceiling of
    $w(F')$, and $w(F') \prec w(F)$.

    \item Otherwise, $w(F)$ and $w(F')$ have the same floor and the
    same ceiling. We apply the same kind of argument as above to the
    left side of the corridor to resolve the order, and if this does
    not work, we use the right side of the corridor to resolve the
    ordering.
\end{compactitem}
\medskip%

Now, it is easy to verify that computing $w(\CurveSet)$ is an \LP-type
problem, with the combinatorial dimension being $\cDim=4$.  Thus, the
maximum corridor in this decomposition is the desired answer to the
point location query, and this now can be solved using this ordering
via an \LP-type solver.

\end{document}